\pdfoutput=1
\documentclass{IEEEtran}
\usepackage{amsthm}
\usepackage{mathtools}
\usepackage[OT2,T1]{fontenc}
\usepackage{hyperref}
\usepackage{cite}
\usepackage{color}
\usepackage{varwidth}

\usepackage{microtype}
\DeclareMicrotypeSet{no-tt}{family={rm*,sf*}}
\microtypesetup{protrusion=no-tt}

\DeclareFontFamily{OT2}{cmr}{\hyphenchar\font45}
\DeclareFontShape{OT2}{cmr}{m}{n}{<->wncyr10}{}

\DeclarePairedDelimiter{\floor}{\lfloor}{\rfloor}
\renewcommand{\c}{,\discretionary{}{}{}}
\newcommand{\Q}{Q_{+}}
\newcommand{\R}{\hat R}
\newcommand{\hattimes}{\mathbin{\hat\times}}
\newcommand{\shuffle}{\mathbin{\text{\fontencoding{OT2}\selectfont x}}}
\newcommand{\new}{\text{new}}
\renewcommand{\d}{\mathrlap{\overline{\color{white}{\vphantom{I}\mathtt{q}}}}}
\DeclareMathOperator{\doub}{d}
\DeclareMathOperator{\n}{n}
\DeclareMathOperator{\rowsum}{sum}
\newcommand*{\email}[1]{\href{mailto:#1}{#1}}

\newtheorem{theorem}{Theorem}
\newtheorem{lemma}{Lemma}
\newtheorem{corollary}{Corollary}
\theoremstyle{definition}
\newtheorem{example}{Example}

\begin{document}
\title{New Infinite Families of Perfect Quaternion Sequences and Williamson Sequences}

\author{Curtis~Bright,
Ilias~Kotsireas~\IEEEmembership{Member,~IEEE},
and~Vijay~Ganesh
\thanks{Manuscript received May 3, 2019; revised July 21, 2020; accepted July 22, 2020. Date of publication August 13, 2020; date of current version
November 20, 2020. \emph{(Corresponding author: Curtis Bright.)}}%
\thanks{Curtis Bright was with the Department of Electrical and Computer Engineering at the University of Waterloo, Waterloo, ON N2L 3G1, Canada.  He is now with the School of Computer Science at the University of Windsor, Windsor, ON N9B 3P4, Canada (e-mail: \email{cbright@uwindsor.ca}).}%
\thanks{Ilias Kotsireas is with the Department of Physics and Computer Science at Wilfrid Laurier University, Waterloo, ON N2L 3C5, Canada (e-mail: \email{ikotsire@wlu.ca}).}%
\thanks{Vijay Ganesh is with the Department of Electrical and Computer Engineering at the University of Waterloo, Waterloo, ON N2L 3G1, Canada (e-mail: \email{vganesh@uwaterloo.ca}).}%
\thanks{Communicated by K.-U. Schmidt, Associate Editor for Sequences.}%
\thanks{Digital Object Identifier 10.1109/TIT.2020.3016510}%
}

\IEEEpubid{\begin{minipage}{\textwidth}\ \\[12pt] \centering%
0018-9448 © 2020 IEEE. Personal use is permitted, but republication/redistribution requires IEEE permission.\\
See \url{https://www.ieee.org/publications/rights/index.html} for more information.
\end{minipage}}

\hypersetup{
hidelinks,
pdfauthor={Curtis Bright, Ilias Kotsireas, Vijay Ganesh},
pdftitle={New Infinite Families of Perfect Quaternion Sequences and Williamson Sequences}}

\maketitle

\begin{abstract}\boldmath
We present new constructions for
perfect and odd perfect sequences over the quaternion group~$Q_8$.
In particular, we show for the first time
that perfect and odd perfect quaternion sequences exist in all lengths~$2^t$ for $t\geq0$.
In doing so we disprove the quaternionic form of Mow's conjecture
that the longest perfect $Q_8$-sequence that can be constructed 
from an orthogonal array construction is of length 64.
Furthermore, we use a connection to combinatorial design theory
to prove the existence of a new infinite class of Williamson sequences,
showing that Williamson sequences of length~$2^t n$ exist for all $t\geq0$
when Williamson sequences of odd length~$n$ exist.  Our constructions
explain the abundance of Williamson sequences in lengths
that are multiples of a large power of two.
\end{abstract}


\section{Introduction}

Sequences that have zero correlation with themselves
after a nontrivial cyclic shift are known as \emph{perfect}~\cite{schmidt2016sequences}.
Such sequences have a long history~\cite{sarwate1979bounds}
and an amazing wealth of applications,
for example, appearing in the 3GPP LTE standard~\cite{park2016optimal}.
Perfect sequences and their generalizations have been applied to
spread spectrum multiple access systems~\cite{suehiro1988modulatable},
radar systems~\cite{kretschmer1991low},
fast start-up equalization~\cite{qureshi1977fast},
channel estimation and synchronization~\cite{gong2013large},
peak-to-average power ratio reduction~\cite{rahmatallah2013peak},
image watermarking~\cite{tirkel1998image}, and
constructing complementary sets~\cite{popovic1990complementary}.
They also tend to have small aperiodic
correlation~\cite{zhang1993polyphase} and small
ambiguity function sidelobes~\cite{popovic1992generalized}.

One of the first researchers to study perfect sequences was
Heimiller~\cite{heimiller1961phase}
who in 1961 gave a construction for perfect sequences
using matrices with array orthogonality.
His construction generated perfect sequences of length
$p^2$ over the complex $p$th roots of unity for any prime $p$.
Shortly after Heimiller's paper was published, Frank
and Zadoff published a response~\cite{frank1962phase}
pointing out that Frank had discovered the same construction
a decade prior as an aircraft
engineer at the Sperry Gyroscope Company.
Moreover, Frank was
granted a patent for a communication system based on
his sequences~\cite{frank1963phase}.

\IEEEpubidadjcol
Frank's construction generates perfect sequences of
length~$n^2$ over the complex $n$th roots of unity.
Given the theoretical elegance of perfect sequences and their
importance to many fields of engineering
it would be extremely interesting and useful if a construction
could be found that produces perfect
sequences longer than $n^2$ using $n$th roots of unity.
However, over 60 years
of effort has failed to find such a construction and
it has been conjectured by
Mow~\cite{mow1996new} that such a construction does not exist.
In light of this, several researchers have searched for perfect sequences
over other alphabets such as the group of quaternions~\cite{kuznetsov2017}.  Quaternions are generalizations
of the complex numbers
that include the additional numbers $j$ and $k$ that satisfy the relationships
\[ i^2=j^2=k^2=-1, \quad ij = k,\quad jk = i,\quad ki = j . \]
Note that these relationships imply that quaternions are generally noncommutative.

Communication systems have been described that are based on the quaternions~\cite{zetterberg1977codes,sethuraman2003full}
and standard mathematical techniques like the Fourier transform have
been generalized to the quaternions for usage in signal processing~\cite{bulow2001hypercomplex} and image processing~\cite{ell2007hypercomplex}.
Perfect sequences over the quaternions were first studied by Kuznetsov~\cite{kuznetsov2009perfect}.
Kuznetsov and Hall~\cite{kuznetsov2010perfect} showed that Mow's conjecture cannot be directly extended to these
sequences by constructing a perfect sequence over a quaternion alphabet
with~$24$ elements and whose length is over 5 billion.
In 2012, Barrera Acevedo and Hall~\cite{acevedo2012perfect} constructed
perfect sequences over the alphabet $\{\pm1,\pm i,j\}$ in all lengths of the form
$q+1$ where $q\equiv1\pmod{4}$ is a prime power.
Most recently, Blake~\cite{blake2016constructions} has run extensive searches for perfect sequences
over the basic quaternion alphabet $Q_8\coloneqq\{\pm1,\pm i,\pm j,\pm k\}$
(see Section~\ref{sec:previouswork} for more on previous work).
The longest perfect $Q_8$-sequence generated from an orthogonal array construction that
Blake found had length $64$.
Intriguingly, this length
is exactly the square of the alphabet size of $8$.
This led Blake to conjecture a quaternionic form of Mow's conjecture
that the quadratic Frank--Heimiller bound applies to perfect sequences
over $Q_8$ generated from an orthogonal array construction.

In this paper we show that Blake's conjecture is false.
This is accomplished via a new construction for perfect quaternion sequences
that can be used to construct arbitrarily long perfect $Q_8$-sequences
(see Section~\ref{sec:perfect}) and odd perfect $Q_8$-sequences (see Section~\ref{sec:nega}).
In particular, we construct the first infinite family of odd perfect
quaternion sequences and show for the first time
that perfect and odd perfect quaternion sequences of length $2^t$ exist for all $t$.
This is starkly different from what happens if one restricts the alphabet
to only include the purely real or complex elements of $Q_8$.
It is known that the longest perfect $\{\pm1\}$-sequence of length $2^t$ is of length four~\cite{turyn1965character}
and the longest perfect $\{\pm1,\pm i\}$-sequence of length $2^t$
is of length sixteen~\cite{turyn1970complex}.
Additionally, the longest odd perfect $\{\pm1\}$-sequence
has length two~\cite{luke2003binary}.

Recently Barrera Acevedo and Dietrich~\cite{acevedo2019new,acevedo2018perfect} discovered a correspondence
between perfect sequences over the quaternions and Williamson sequences
from combinatorial design theory~\cite{williamson1944hadamard} (see Section~\ref{sec:preliminaries}
for the definition of Williamson sequences).
Using this relationship we construct new
Williamson sequences in even lengths, including all powers of two.
Prior to our construction Williamson sequences of length~$2^t$
were only known to exist for $t\leq6$~\cite{bright2018sat}.
See Section~\ref{sec:construction} for our Williamson sequence construction
and Section~\ref{sec:nega} for a construction
for a variant of Williamson sequences that
our Williamson sequence construction relies on.

In 1944, when Williamson introduced the sequences that now bear his name~\cite{williamson1944hadamard}
he showed that the existence of Williamson sequences of odd length $n$
implied the existence of Williamson sequences of length $2n$.
In 1970, twenty-six years later, Turyn~\cite{turyn1970complex} generalized Williamson's result by showing that
the existence of Williamson sequences of odd length $n$ implied
the existence of Williamson sequences of lengths~$2^t n$ for $t\leq4$.
Nearly fifty years after Turyn's result, Barrera Acevedo and Dietrich~\cite{acevedo2019new} improved this to $t\leq6$ in 2019.
In this paper we complete this process of generalizing Williamson's doubling result
by showing the result in fact holds for all $t\geq0$.
In other words, the existence of Williamson sequences of odd length~$n$
implies the existence of Williamson sequences of length $2^t n$ for all $t$
and this provides a large new infinite class of Williamson sequences.

Exhaustive searches for Williamson sequences~\cite{holzmann2008williamson,bright2019applying} have shown that
they exist in all lengths $n<65$ except for $n=35$, $47$, $53$, and $59$.  They are generally
very abundant in the even lengths (particularly
in the lengths that are divisible by a large power of two) but not in the odd lengths.
For example, over 50,000 inequivalent sets of Williamson sequences
exist in length~$64$ but fewer than 100 sets of Williamson
sequences exist in all odd lengths up to $65$.
Previously this dichotomy was unexplained
but the constructions that we provide in this paper produce
approximately 75\% of the Williamson sequences that exist in all even lengths $n\leq70$
(see Section~\ref{sec:conclusion}).

\section{Preliminaries}\label{sec:preliminaries}

In this section we provide the preliminaries necessary to explain our
construction for perfect quaternion sequences, odd perfect quaternion sequences, and Williamson sequences.
First we define the concept of sequence perfection in terms of the amount of
correlation that a sequence has with cyclically shifted copies of itself.
Following this, we discuss the array orthogonality property used in several
constructions for perfect sequences.
We then define Williamson sequences and finally present Barrera Acevedo and Dietrich's
equivalence between perfect quaternion sequences and Williamson sequences.

\subsection{Complementary sequences and perfect sequences}

The \emph{aperiodic} crosscorrelation of two sequences $A=[a_0,\dotsc,a_{n-1}]$
and $B=[b_0,\dotsc,b_{n-1}]$ of length~$n$ is given by
\[ C_{A,B}(t) \coloneqq \sum_{r=0}^{n-t-1} a_r b_{r+t}^*
\]
where $z^*$ denotes the conjugate of $z$~\cite{sarwate1980crosscorrelation}.
For a general quaternionic number we have
$(a+bi+cj+dk)^*\coloneqq a-bi-cj-dk$ where $a$, $b$, $c$, $d$ are real.
The aperiodic autocorrelation of $A$
is given by $C_A(t)\coloneqq C_{A,A}(t)$.
More generally, we also define the \emph{periodic} and \emph{odd periodic} (or \emph{negaperiodic})
crosscorrelations by
\begin{align*}
R_{A,B}(t) &\coloneqq C_{A,B}(t)+C_{B,A}(n-t)^* &&\text{and} \\
\R_{A,B}(t) &\coloneqq C_{A,B}(t) - C_{B,A}(n-t)^* &&\text{for $0\leq t<n$}.
\end{align*}
The periodic and odd periodic autocorrelations
of $A$ are given by $R_A(t)\coloneqq R_{A,A}(t)$ and $\R_A(t)\coloneqq\R_{A,A}(t)$ respectively.
These functions may be extended to all integers $t$ via the expressions
\begin{gather*}
R_{A,B}(t) \coloneqq \sum_{r=0}^{n-1} a_r b^*_{r+t\bmod n}
\intertext{and}
\R_{A,B}(t) \coloneqq \sum_{r=0}^{n-1}(-1)^{\floor{(r+t)/n}}a_r b^*_{r+t\bmod n} .
\end{gather*}

A set
$S$ of sequences of length~$n$
is called \emph{complementary} if $\sum_{A\in S} C_{A}(t)=0$ for all \mbox{$1\leq t<n$}.\footnote{Technically $S$ should be defined to be a multiset but we follow standard convention and refer to it as a set.}
Similarly, $S$ is called \emph{periodic complementary} if $\sum_{A\in S} R_{A}(t)=0$
and \emph{odd periodic complementary} (or \emph{negacomplementary}) if $\sum_{A\in S}\R_{A}(t)=0$ for all $1\leq t<n$.
A single sequence $A$ is called \emph{perfect} if $R_A(t)=0$ for all $1\leq t<n$
and \emph{odd perfect} if $\R_A(t)=0$ for all $1\leq t<n$.

Note that the periodic correlation values $R_A(t)$ of a sequence are preserved under
the \emph{cyclic shift} operator $[a_0,\dotsc,a_{n-2},a_{n-1}]\mapsto[a_{n-1},a_0,\dotsc,a_{n-2}]$
and the negaperiodic correlation values $\R_A(t)$ of a sequence are preserved under the
\emph{negacyclic shift} operator
$[a_0,\dotsc,a_{n-2},a_{n-1}]\mapsto[-a_{n-1},a_0,\dotsc,a_{n-2}]$ (see~\cite{bomer1990periodic,parker2001even}).
Therefore applying a cyclic shift to any sequence in a set of periodic complementary sequences
or applying a negacyclic shift to any sequence in a set of negacomplementary sequences
does not disturb the property of the set being periodic complementary or negacomplementary.

Let $(-1)*X$ denote the sequence whose $r$th entry is $(-1)^r x_r$, i.e., the \emph{alternating
negation} operation.  A set of periodic complementary sequences of odd length $n$
can be converted into a set of negacomplementary sequences and vice versa by applying
the alternating negation operation (see~\cite{xia2006hadamard}).

\begin{lemma}\label{lem:nega}
If $n$ is odd then $(A,B,C,D)$ are periodic complementary sequences of length $n$
if and only if $((-1)*A, (-1)*B,(-1)*C, (-1)*D)$ are negacomplementary sequences.
\end{lemma}

\begin{example}
$(\verb|+--|, \verb|+--|, \verb|+--|, \verb|+++|)$ is a set of complementary sequences of length~$3$
and $(\verb|++-|, \verb|++-|, \verb|++-|, \verb|+-+|)$ is the set of negacomplementary sequences
generated from it using Lemma~\ref{lem:nega}.  Note that we follow the convention of writing $1$s
by $\verb|+|$ and $-1$s by $\verb|-|$.
\end{example}

\subsection{Matrices with array orthogonality}

The sequences in a set $S$ are \emph{periodically uncorrelated} if
any two distinct sequences $A$, $B$ in $S$ satisfy $R_{A,B}(t)=0$ for all $0\leq t<n$.
If a set of periodic complementary sequences $S=\{S_1,\dotsc,S_m\}$ (with each sequence of length~$n$ a multiple of $m$)
are periodically uncorrelated then the $n\times m$ matrix whose columns are given by
$S_1$, $\dotsc$, $S_m$ is said to have \emph{array orthogonality}.

Matrices with array orthogonality are important because they are used
in many constructions for perfect sequences such as in the Frank--Heimiller and Mow constructions.
In particular, the sequence formed
by concatenating the rows of a matrix with array orthogonality is perfect.  Additionally,
matrices with array orthogonality are themselves perfect arrays.
\emph{Perfect arrays} are often studied as a way of generalizing the concept of perfection 
from sequences to matrices (e.g., see~\cite{bomer1990two,fan1995synthesis}).
They are defined to be $n\times m$ matrices $A=(a_{r,s})$ that satisfy
\[ \sum_{r=0}^{n-1}\sum_{s=0}^{m-1}a_{r,s}a^*_{r+t\bmod n,s+t'\bmod n} = 0 \]
for all $(t,t')\neq(0,0)$ with $0\leq t<n$, $0\leq t'<m$.

\begin{example}
The columns of the matrix
\[ \begin{bmatrix}
\verb|+| & \verb|+| \\
\verb|+| & \verb|-|
\end{bmatrix} \]
have array orthogonality and therefore this is a perfect array.  It generates the perfect sequence $[\verb|+++-|]$.
\end{example}

\subsection{Williamson and nega Williamson sequences}

First, we describe the symmetry properties that are used in the definition
of Williamson sequences and will be important in our construction for Williamson sequences.
A sequence $A$ of length $n$ is called \emph{symmetric} if $a_t=a_{n-t}$ for all $1\leq t<n$
and is \emph{palindromic} if $a_t=a_{n-t-1}$ for all $0\leq t<n$.
For example, $[x,y,z,y]$ is symmetric and $[y,z,y]$ is palindromic.
Note that a sequence is symmetric if and only
if the subsequence formed by removing its first element is palindromic.
Additionally, we call a sequence \emph{antipalindromic} if $a_t=-a_{n-t-1}$ for all $0\leq t<(n-1)/2$
and \emph{antisymmetric} if $a_t=-a_{n-t}$ for all $1\leq t<n/2$.
If $[X;Y]$ denotes sequence concatenation and~$\tilde X$
denotes the reversal of $X$ then $[X;\tilde X]$ is palindromic
and $[X;-\tilde X]$ is antipalindromic.

A quadruple of symmetric $\{\pm1\}$-sequences $(A,B,C,D)$ are known as \emph{Williamson}
if they are periodic complementary, i.e., if $R_A(t)+R_B(t)+R_C(t)+R_D(t)=0$ for all $1\leq t<n$.
Additionally, we call a quadruple of $\{\pm1\}$-sequences $(A,B,C,D)$ \emph{nega Williamson}
if they are negacomplementary, i.e., $\R_A(t)+\R_B(t)+\R_C(t)+\R_D(t)=0$ for all $1\leq t<n$.

As is conventional, we require by definition that Williamson sequences are symmetric.
If the symmetry condition is replaced with the weaker property that the sequences
are instead \emph{amicable} (i.e., $R_{X,Y}(t)=R_{Y,X}(t)$ for all $t$ and $X,Y\in\{A,B,C,D\}$) then the
sequences are known as \emph{Williamson-type}.
For nega Williamson sequences we do not require them to be symmetric.
Instead, our work
has discovered the importance of \emph{palindromic} and \emph{antipalindromic} nega Williamson sequences;
see Section~\ref{sec:construction} for details.

\begin{example}\label{ex:simple}
$(\verb|++|,\verb|++|,\verb|+-|,\verb|+-|)$ are Williamson sequences of length~$2$.
Similarly, $(\verb|+-|,\verb|+-|,\verb|+-|,\verb|+-|)$ are antipalindromic nega Williamson sequences
and $(\verb|++|,\verb|++|,\verb|++|,\verb|++|)$ are palindromic nega Williamson sequences.
$(\verb|++-+|,\verb|++-+|,\verb|++-+|,\verb|++-+|)$ are Williamson sequences of length~$4$.
Similarly, $(\verb|+-+-|,\verb|+-+-|,\verb|++--|,\verb|++--|)$ are antipalindromic nega Williamson sequences
and $(\verb|+--+|,\verb|+--+|,\verb|++++|,\verb|++++|)$ are palindromic nega Williamson sequences.
\end{example}

\subsection{The Barrera Acevedo--Dietrich correspondence}

Let $Q_8\coloneqq\{\pm1,\pm i,\pm j,\pm k\}$ be the quaternion group
and $\Q\coloneqq Q_8\cup qQ_8$ where $q\coloneqq(1+i+j+k)/2$
(note that $\Q$ is a set of sixteen quaternions that is \emph{not} a group).
Barrera Acevedo and Dietrich~\cite[Theorem~2.4]{acevedo2019new} show that there is an equivalence between
Williamson-type sequences and perfect sequences over $\Q$.

\begin{theorem}\label{thm:corr}
There is a one-to-one correspondence between sets of Williamson-type sequences of length~$n$ and
perfect sequences of length~$n$ over $\Q$.
\end{theorem}

This correspondence is made explicit through the following mapping between the $r$th
entries of a set of Williamson-type sequences $(a_r,b_r,c_r,d_r)$ and the $r$th entry
of the corresponding perfect sequence $s_r$:
\[
\begin{matrix}
a_r & \verb|-| &  \verb|+| &  \verb|+| &  \verb|+| &  \verb|+| &  \verb|+| &  \verb|+| &  \verb|+|  \\
b_r & \verb|-| & \verb|-| &  \verb|+| & \verb|-| & \verb|-| &  \verb|+| &  \verb|+| & \verb|-|  \\
c_r & \verb|-| & \verb|-| & \verb|-| &  \verb|+| & \verb|-| & \verb|-| &  \verb|+| &  \verb|+|  \\
d_r & \verb|-| &  \verb|+| & \verb|-| & \verb|-| & \verb|-| &  \verb|+| & \verb|-| &  \verb|+|  \\ \hline
s_r &  1 &  i &  j &  k &  q & qi & qj & qk 
\end{matrix}
\]
Additionally, we have the rule that if $(a_r,b_r,c_r,d_r)$ maps to~$s_r$ then $(-a_r,-b_r,-c_r,-d_r)$ maps to $-s_r$.

Because of this theorem a construction for Williamson sequences of length~$n$ also
produces symmetric perfect $\Q$\nobreakdash-sequences.  In this paper
we state our construction in terms of Williamson sequences with the understanding
that it equivalently produces perfect $\Q$-sequences.
In Section~\ref{sec:perfect} we also show how our construction can be used
to produce perfect $Q_8$-sequences in many lengths including all powers of two.

If the entries of a set of Williamson sequences $(A,B,C,D)$
of length $n$ satisfy $a_rb_rc_rd_r=1$ for all $0\leq r<n$ then
the Barrera Acevedo--Dietrich correspondence produces perfect $Q_8$-sequences.
For this reason we say that a quadruple of sequences has the \emph{$Q_8$-property}
if the entries of its sequences satisfy $a_rb_rc_rd_r=1$ for all $0\leq r<n$.
In his original paper~\cite{williamson1944hadamard} Williamson proved that the entries of
all Williamson sequences in odd lengths~$n$ satisfy $a_rb_rc_rd_r=-a_0b_0c_0d_0$
for $1\leq r<n$.  As a consequence, no Williamson sequence
of odd length $n>1$ can have the $Q_8$-property.  However, many Williamson
sequences in even lengths have the $Q_8$-property.  Exhaustive searches~\cite{bright2019applying}
have found Williamson sequences with the $Q_8$-property
in all even lengths $n\leq2^5$ except for $6$, $12$, and $28$.

Barrera Acevedo and Dietrich also use the periodic product construction of L\"uke~\cite{luke1988sequences} that
generates perfect sequences from shorter perfect sequences.
Let $X\times Y$ be the sequence whose $r$th entry is $x_{r\bmod n}y_{r\bmod m}$
for $0\leq r<nm$ (where $X$ has length $n$ and $Y$ has length $m$).

\begin{theorem}\label{thm:comp}
Suppose $X$ and $Y$ have coprime lengths $n$ and~$m$.
If $X$ is a perfect $Q_8$-sequence and $Y$ is a perfect $\Q$\nobreakdash-sequence then
$X\times Y$ is a perfect $\Q$-sequence.
Furthermore, if $X$ and $Y$ are symmetric then $X\times Y$ is symmetric.
\end{theorem}

Theorems~\ref{thm:corr} and~\ref{thm:comp} immediately yield the following corollary.

\begin{corollary}
If a perfect symmetric $Q_8$-sequence of length~$2^t$ exists and Williamson sequences
of odd length $n$ exist then Williamson sequences of length $2^tn$ exist.
\end{corollary}

Barrera Acevedo and Dietrich also provide examples of symmetric perfect $Q_8$-sequences
for all $t\leq6$.
In this paper we extend their result by showing
that perfect symmetric $Q_8$\nobreakdash-sequences of length $2^t$
exist for all $t\geq0$ and therefore show
Williamson sequences exist in all lengths of the form~$2^t n$
whenever Williamson sequences exist in odd length~$n$.

\begin{example}
Using the Barrera Acevedo--Dietrich correspondence the Williamson sequences
$(\verb|++-+|,\verb|++-+|,\verb|++-+|,\verb|++-+|)$ produce
the perfect $Q_8$-sequence $[\verb|--+-|]$, the Williamson sequences
$(\verb|++--+|,\verb|-+--+|,\verb|-++++|,\verb|-++++|)$
produce the perfect $\Q$-sequence $[\verb|q-jj-|]$
and the Williamson sequences
$(\verb|++--+--+|,\verb|++--+--+|,\verb|+++-+-++|,\verb|+++-+-++|)$
produce the perfect $Q_8$-sequence $[\verb|--j+-+j-|]$.
(We denote $i$, $j$, $k$, and $q$ by $\verb|i|$,
$\verb|j|$, $\verb|k|$, and $\verb|q|$.)
\end{example}

\section{Constructions for Williamson sequences}\label{sec:construction}

Our main construction is based on the following simple sequence operations.
\begin{enumerate}
\item The \emph{doubling} of $X$, denoted by $\doub(X)$, i.e., $\doub(X)\coloneqq[x_0,\dotsc,x_{n-1},x_0,\dotsc,x_{n-1}]$.
\item The \emph{negadoubling} of $X$, denoted by $\n(X)$, i.e., $\n(X)\coloneqq[x_0,\dotsc,x_{n-1},-x_0,\dotsc,-x_{n-1}]$.
\item The \emph{interleaving} of $X$ and $Y$, denoted by $X\shuffle Y$, i.e., $X\shuffle Y\coloneqq[x_0,y_0,x_1,y_1,\dotsc,x_{n-1},y_{n-1}]$.
\end{enumerate}

We will use the following properties of these operations in our construction.
For completeness, proofs of these properties are given in the appendix.  In each property~$X$ and~$Y$ are
arbitrary sequences of the same length and~$t$ is an arbitrary integer.
\begin{enumerate}
\item $R_{\doub(X)}(t)=2R_X(t)$.
\item $R_{\n(X)}(t)=2\R_X(t)$.
\item $R_{\doub(X),\n(Y)}(t)=0$ and $R_{\n(Y),\doub(X)}(t)=0$.
\item $R_{X\shuffle Y}(2t)=R_X(t)+R_Y(t)$.
\item $R_{X\shuffle Y}(2t+1)=R_{X,Y}(t)+R_{Y,X}(t+1)$.
\item If $X$ is symmetric then $\doub(X)$ is symmetric.
\item If $X$ is antipalindromic and of even length then $\n(X)$ is palindromic.
\item If $X$ is symmetric and $Y$ is palindromic then $X\shuffle Y$ is symmetric.
\end{enumerate}

Our main construction for Williamson sequences is given by the following theorem.

\begin{theorem}\label{thm:main}
If $(A,B,C,D)$ are Williamson sequences of even length $n$ and $(A',B',C',D')$
are antipalindromic nega Williamson sequences of length $n$ then
\[ (\doub(A)\shuffle\n(A'),\doub(B)\shuffle\n(B'),\doub(C)\shuffle\n(C'),\doub(D)\shuffle\n(D')) \]
are Williamson sequences of length $4n$.
\end{theorem}

This theorem can be used to construct a large number of new Williamson sequences, assuming that
sequences that satisfy the preconditions are known.
For example, if~$N$ Williamson sequences
of length~$n$ are known and~$M$ antipalindromic nega Williamson sequences of length~$n$ are known
this theorem immediately implies that at least~$NM$ Williamson sequences of length~$4n$ exist.

We now prove Theorem~\ref{thm:main}.

\begin{proof}
Let $X_\new\coloneqq\doub(X)\shuffle\n(X')$ for $X\in\{A,B,C,D\}$.
By construction it is clear that $X_\new$ is a $\{\pm1\}$-sequence of length $4n$.
Additionally, $X_\new$ is symmetric by property 8 since $\doub(X)$ is symmetric by property 6
and $\n(X')$ is palindromic by property 7.
It remains to show that
\[ R_{A_\new}(t) + R_{B_\new}(t) + R_{C_\new}(t) + R_{D_\new}(t) = 0 \text{ for $1\leq t<4n$} . \]
Note that by properties 4 and 5 we have $R_{X_\new}(t)$ is
\[ \begin{cases}
R_{\doub(X)}(t/2)+R_{\n(X')}(t/2) & \text{if $t$ is even,} \\
R_{\doub(X),\n(X')}\bigl(\frac{t-1}{2}\bigr) + R_{\n(X'),\doub(X)}\bigl(\frac{t+1}{2}\bigr) & \text{if $t$ is odd} .
\end{cases} \]
By property 3 we have $R_{X_\new}(t)=0$ when $t$ is odd.
When $t$ is even by properties 1 and~2 we have
\[ R_{X_\new}(t) = 2R_X(t/2) + 2\R_{X'}(t/2) . \]
When $t=2n$ we have $R_X(t/2)=n$ and $\R_{X'}(t/2)=-n$ (because $X$ and $X'$ both have length $n$) so $R_{X_\new}(t)=0$ in this case.  Otherwise
we have
\[ \sum_{X=A,B,C,D}R_{X}(t/2) = 0 \quad\text{and}\quad \sum_{X=A,B,C,D}\R_{X'}(t/2) = 0 \]
for even $t\neq2n$ with $1\leq t<4n$
since $(A,B,C,D)$ are Williamson sequences and $(A',B',C',D')$ are nega Williamson sequences.
It follows that $\sum_{X=A,B,C,D}R_{X_\new}(t)=0$, as required.
\end{proof}

\begin{example}
Using the set of Williamson sequences $(\verb|++-+|,\verb|++-+|,\verb|++-+|,\verb|++-+|)$
and the set of antipalindromic nega Williamson sequences $(\verb|+-+-|,\verb|+-+-|,\verb|++--|,\verb|++--|)$
in Theorem~\ref{thm:main} produces the set of Williamson sequences
\begin{center} $\mathllap{(}\verb|+++--++-+-++--++|,\verb|+++--++-+-++--++|\mathrlap{,}$ \\
$\verb|++++--+-+-+--+++|,\verb|++++--+-+-+--+++|\mathrlap{).}$ \end{center}
\end{example}

Note that the assumption that $n$ is even is essential to the theorem.  If $n$ is odd
the constructed sequences will not be symmetric.
Additionally, antipalindromic nega Williamson sequences
do not exist in odd lengths $n>1$, as we now show.

\begin{lemma}
Antipalindromic nega Williamson sequences do not exist in odd lengths except for $n=1$.
\end{lemma}

\begin{proof}
Let $(A,B,C,D)$ be a hypothetical set of antipalindromic nega Williamson sequences of odd length~$n$.
Note that a sequence $X$ of odd length $n$ is antipalindromic if and only if $X'\coloneqq(-1)*X$ is antipalindromic.
By Lemma~\ref{lem:nega}, it follows that
$(A',B',C',D')$ are antipalindromic periodic complementary sequences.
If $\rowsum(X)$ denotes the row sum of $X$ it is well-known
(e.g., via the Wiener--Khinchin theorem or sequence compression~\cite{dokovic2015compression})
that
\[ \rowsum(A')^2 + \rowsum(B')^2 + \rowsum(C')^2 + \rowsum(D')^2 = 4n . \]
However, if $X'$ is antipalindromic we have $\rowsum(X')=\pm1$ and the above sum of squared row sums
must be equal to four, implying that $n=1$.
\end{proof}

Although the construction in Theorem~\ref{thm:main} requires
sequences of even lengths there is a variant of the construction
that uses sequences of odd lengths.
The proof is similar and uses the following additional properties.
\begin{enumerate}
\item[9)] If $X$ is antisymmetric and of odd length then $\n(X)$ is symmetric.
\item[10)] If $X$ is palindromic then $\doub(X)$ is palindromic.
\end{enumerate}

\begin{theorem}\label{thm:altmain}
Let $n$ be odd and let $(A,B,C,D)$ and $(A',B',C',D')$ be quadruples of sequences of length~$n$.
If $(A,B,C,D)$ is the result of applying $(n-1)/2$ cyclic shifts to each member in a set of Williamson sequences
and $(A',B',C',D')$ is the result of applying $(n+1)/2$ negacyclic shifts to each member in a
set of palindromic nega Williamson sequences then
\[ (\n(A')\shuffle\doub(A),\n(B')\shuffle\doub(B),\n(C')\shuffle\doub(C),\n(D')\shuffle\doub(D)) \]
are Williamson sequences of length~$4n$.
\end{theorem}

\begin{proof}
Since a symmetric sequence of odd length $n$ becomes palindromic after $(n-1)/2$ cyclic shifts,
$(A,B,C,D)$ are palindromic.  Since a palindromic sequence of odd length~$n$ becomes antisymmetric
after $(n+1)/2$ negacyclic shifts, $(A',B',C',D')$ are antisymmetric.

Let $X_\new\coloneqq\n(X')\shuffle\doub(X)$ for $X\in\{A,B,C,D\}$.
By construction it is clear that $X_\new$ is a $\{\pm1\}$-sequence of length~$4n$.
Additionally, $X_\new$ is symmetric by property 8 since $\n(X')$ is symmetric by property 9
and $\doub(X)$ is palindromic by property 10.
The remainder of the proof now proceeds as in the proof of Theorem~\ref{thm:main}.
\end{proof}

\begin{example}
Note $(\verb|++--+|,\verb|-+--+|,\verb|-++++|,\verb|-++++|)$
are Williamson sequences and $(\verb|+---+|,\verb|++-++|,\verb|+---+|,\verb|+++++|)$
are palindromic nega Williamson sequences.
Then using $(A,B,C,D)=(\verb|-+++-|,\verb|-+-+-|,\verb|++-++|,\verb|++-++|)$
and $(A',B',C',D')=(\verb|++-+-|,\verb|+--++|,\verb|++-+-|,\verb|---++|)$
in Theorem~\ref{thm:altmain} produces the following Williamson sequences of length~$20$:
\begin{center} $\mathllap{(}\verb|+-++-+++-----+++-++-|,\verb|+--+--+++---+++--+--|\mathrlap{,}$ \\
$\verb|++++--++-+-+-++--+++|,\verb|-+-+--+++++++++--+-+|\mathrlap{).}$ \end{center}
\end{example}

\section{Nega Williamson sequences and odd perfect sequences}\label{sec:nega}

Although antipalindromic nega Williamson sequences do not exist
in odd lengths larger than~$1$ we now present constructions
showing that they exist in many even lengths.  First, note that
when the length is even there is a correspondence between palindromic and antipalindromic nega Williamson sequences.
In other words, to use Theorem~\ref{thm:main} it is sufficient to find palindromic nega Williamson sequences.

\begin{lemma}\label{lem:pal}
There is a one-to-one correspondence between
palindromic and antipalindromic nega Williamson sequences in even lengths.
\end{lemma}

\begin{proof}
Applying $n/2$ negacyclic shifts to each sequence in a set
of palindromic nega Williamson sequences of even length~$n$
yields a set of antipalindromic nega Williamson sequences.  Similarly, a set of
palindromic nega Williamson sequences can be produced from a set
of antipalindromic nega Williamson sequences by applying the negacyclic shift operator
$n/2$ times.
\end{proof}

\begin{example}
$(\verb|+--++-|,\verb|+-+-+-|,\verb|++-+--|,\verb|+++---|)$ are antipalindromic
nega Williamson sequences
that may be converted into
the palindromic nega Williamson sequences
$(\verb|--++--|,\verb|+-++-+|,\verb|-++++-|,\verb|++++++|)$ using the translation
in Lemma~\ref{lem:pal}.
\end{example}

We now provide a construction that shows that infinitely many
palindromic nega Williamson sequences exist.
Recall that $\tilde X$ denotes the reverse of the sequence $X$ and $[X;Y]$ denotes the concatenation
of~$X$ and~$Y$.
Note that if~$X$ and~$Y$ have length $n$ then $C_{[X;Y]}(t)=C_X(t)+C_Y(t)+C_{Y,X}(n-t)^*$
and $C_{[X;Y]}(2n-t)=C_{X,Y}(n-t)$ for $0\leq t\leq n$.
First we prove a simple lemma.
\begin{lemma}\label{lem:negadoub}
Applying the negadoubling operator to
the sequences in a complementary set produces a negacomplementary set.
\end{lemma}
\begin{proof}
Note that if $X$ is a sequence of length~$n$ then $\R_{\n(X)}(t)=2C_X(t)$ for $0\leq t\leq n$;
by definition we have $\R_{\n(X)}(t)=C_{\n(X)}(t)-C_{\n(X)}(2n-t)^*$ and when $0\leq t\leq n$ we
have $C_{\n(X)}(t)=C_X(t)+C_{-X}(t)+C_{-X,X}(n-t)^*=2C_X(t)-C_X(n-t)^*$ and $C_{\n(X)}(2n-t)=C_{X,-X}(n-t)=-C_{X}(n-t)$.
Thus $\R_{\n(X)}=2C_X(t)-C_X(n-t)^*+C_{X}(n-t)^*=2C_X(t)$.

Suppose $S$ is a complementary set of sequences of length~$n$.
Using the above property we have
$\sum_{A\in S}\R_{\n(A)}(t)=2\sum_{A\in S}C_A(t)$ for all $0\leq t\leq n$.
Since $S$ is complementary this implies $\sum_{A\in S}\R_{\n(A)}(t)=0$ for all $1\leq t\leq n$.
Using the symmetry $\R_X(t)=\R_X(-t)^*$ shows that this also holds for all $n<t<2n$.
\end{proof}

\begin{example}
Using the set of four complementary sequences $(\verb|+++|,\verb|+--|, \verb|+-+|, \verb|++-|)$
of length three with Lemma~\ref{lem:negadoub} produces the set of four negacomplementary sequences
$(\verb|+++---|, \verb|+---++|, \verb|+-+-+-|, \verb|++---+|)$ of length six.
\end{example}

\begin{theorem}\label{thm:negcon}
If $A$ and $B$ are complementary $\{\pm1\}$-sequences (i.e., $(A,B)$ is a Golay pair) then
\[ ([A;B;\tilde B;\tilde A], [\tilde B;\tilde A;A;B], [-\tilde B;\tilde A;A;-B], [-A;B;\tilde B;-\tilde A]) \]
and
\begin{gather*}
([A;B;\tilde B;\tilde A;A;B;\tilde B;\tilde A], \\
[A;B;-\tilde B;-\tilde A;-A;-B;\tilde B;\tilde A], \\
[A;-B;\tilde B;-\tilde A;-A;B;-\tilde B;\tilde A], \\
[A;-B;-\tilde B;\tilde A;A;-B;-\tilde B;\tilde A])
\end{gather*}
are sets of palindromic nega Williamson sequences with the $Q_8$-property.
\end{theorem}

\begin{proof}
The fact that the sequences are palindromic is immediate from the manner in which they were constructed.
Note that $(\tilde A,\tilde B)$ is a Golay pair since $(A,B)$ is a Golay pair.  Thus $(A,B,\tilde B,\tilde A)$
is a complementary set and since
\[ \begin{bmatrix}
\verb|+| & \verb|+| & \verb|+| & \verb|+| \\
\verb|+| & \verb|+| & \verb|-| & \verb|-| \\
\verb|+| & \verb|-| & \verb|+| & \verb|-| \\
\verb|+| & \verb|-| & \verb|-| & \verb|+|
\end{bmatrix} \] is an orthogonal matrix the sequences
\[ 
\begin{gathered}([A;B;\tilde B;\tilde A], [A;B;-\tilde B;-\tilde A],\\
[A;-B;\tilde B;-\tilde A], [A;-B;-\tilde B;\tilde A])\end{gathered} \tag{$*$}\label{eq:star} \]
form a complementary set of sequences by~\cite[Thm.~7]{tseng1972complementary}.
Since they are complementary they are also negacomplementary
and applying $2n$ negacyclic shifts (where $A$ and $B$ are of length~$n$)
to the second and third sequences and negating the fourth shows the first set in the theorem is negacomplementary.

Since~\eqref{eq:star} are complementary, by Lemma~\ref{lem:negadoub}
applying the negadoubling operator to these sequences will produce negacomplementary sequences.
In other words,
\begin{gather*}
([A;B;\tilde B;\tilde A;-A;-B;-\tilde B;-\tilde A], \\
[A;B;-\tilde B;-\tilde A;-A;-B;\tilde B;\tilde A], \\
[A;-B;\tilde B;-\tilde A;-A;B;-\tilde B;\tilde A], \\
[A;-B;-\tilde B;\tilde A;-A;B;\tilde B;-\tilde A])
\end{gather*}
are negacomplementary.  Applying $4n$ negacyclic shifts (where $A$ and $B$ are of length~$n$) to the
first and last sequences shows the second set in the theorem is negacomplementary.

Lastly, we show that the produced sequences $X$, $Y$, $U$, and~$V$ have the $Q_8$-property.
In the first set we have $x_r=a_r$, $y_r=b_{n-r+1}$, $u_r=-y_r$, and $v_r=-x_r$ for all $0\leq r<n$
and in the second set we have $x_r=y_r=u_r=v_r=a_r$ for all $0\leq r<n$.
In each case $x_ry_ru_rv_r=1$ for all $0\leq r<n$ and more generally we have
$x_ry_ru_rv_r=1$ for all $0\leq r<4n$ in the first set (and $0\leq r<8n$ in the second set).
\end{proof}

\begin{example}
Using the Golay pair $(\verb|++|,\verb|+-|)$
the first set of palindromic nega Williamson sequences
generated by Theorem~\ref{thm:negcon} is $(\verb|+++--+++|,\verb|-++++++-|,\verb|+-++++-+|,\verb|--+--+--|)$;
the second set is
$(\verb|+++--++++++--+++|,\verb|+++-+------+-+++|\c\verb|++-+-+----+-+-++|,\verb|++-++-++++-++-++|)$.
\end{example}

Golay sequences, originally defined by Golay~\cite{golay1961complementary},
are known to exist in lengths $2$, $10$, and $26$~\cite{davis1999peak}.
Since Turyn has shown that Golay pairs in lengths $n$ and $m$ can be composed to
form Golay pairs in length $nm$~\cite{turyn1974hadamard} they also exist 
in all lengths of the form $2^a5^b13^c$ with $a\geq b+c$.
Thus, Theorem~\ref{thm:negcon} implies that palindromic nega Williamson sequences with the $Q_8$-property
exist in all lengths of the form $2^a5^b13^c$ with $a\geq b+c+2$.
In particular, taking $b=c=0$ gives that
palindromic nega Williamson sequences with the $Q_8$-property
exist in all lengths that are powers of two (since palindromic nega Williamson sequences with the $Q_8$-property
of length~$2$ exist, see Example~\ref{ex:simple}).
These are not the only lengths in which palindromic nega Williamson sequences exist, however.
Palindromic nega Williamson sequences in many other lengths may be constructed using
L\"uke's product construction with odd perfect sequences.

First, note that in odd lengths there is an equivalence between Williamson sequences
and palindromic nega Williamson sequences.

\begin{lemma}\label{lem:negacor}
There is a one-to-one correspondence between
palindromic nega Williamson sequences and Williamson sequences in odd lengths.
\end{lemma}

\begin{proof}
Let $(A,B,C,D)$ be a set of Williamson sequences of odd length~$n$.
By Lemma~\ref{lem:nega}, $((-1)*A,(-1)*B,(-1)*C,(-1)*D)$ will be
a set of nega Williamson sequences.  The sequences in this set will be antisymmetric
since if $X$ is symmetric and of odd length then $(-1)*X$ is antisymmetric.
Applying $(n-1)/2$ negacyclic shifts to each sequence 
in this set produces a set of palindromic nega Williamson sequences.
Similarly, arbitrary palindromic nega Williamson sequences of odd length
can be transformed into Williamson sequences by applying the inverse of the above transformations.
\end{proof}

\begin{example}
The set of Williamson sequences
$(\verb|-++--++|\c\verb|---++--|\c\verb|-+----+|\c\verb|-+----+|)$
generates the set of palindromic nega Williamson sequences
$(\verb|++---++|\c\verb|--+-+--|\c\verb|+-----+|\c\verb|+-----+|)$ and vice versa
using the transformation in Lemma~\ref{lem:negacor}.
\end{example}

Since Williamson sequences are known to exist for all lengths $n<35$
this implies that palindromic nega Williamson sequences exist for all
odd lengths up to~$33$.

Next, we note that odd perfect sequences in even lengths can be
composed with odd perfect sequences in odd lengths to generate
longer odd perfect sequences.
Let $X\hattimes Y$ be the sequence whose $r$th entry is $(-1)^{\floor{r/n}+\floor{r/m}}x_{r\bmod n}y_{r\bmod m}$
for $0\leq r<nm$ (where~$X$ has length~$n$ and~$Y$ has length~$m$).

\begin{lemma}\label{lem:oddprod}
Suppose $X$ and $Y$ have coprime lengths $n$ and~$m$, one of which is even.
If $X$ is odd perfect and $Y$ is odd perfect then $X\hattimes Y$
is odd perfect.  Furthermore, if $X$ and $Y$ are palindromic then
$X\hattimes Y$ is antipalindromic.
\end{lemma}

\begin{proof}
The fact that $X\hattimes Y$ is odd perfect follows from
$\R_{X\hattimes Y}(t)=\R_{X}(t)\R_{Y}(t)$ as given in~\cite[Eq.~9]{luke2003binary}.
Suppose that $0\leq r<nm$ is arbitrary.  Since $X$ and $Y$ are palindromic we have
$x_{nm-r-1\bmod n}y_{nm-r-1\bmod m} = x_{n-r-1\bmod n}y_{m-r-1\bmod m}
= x_{r\bmod n}y_{r\bmod m}$.
Thus the $(nm-r-1)$th entry of $X\hattimes Y$ is
\[
\begin{split} (-1)^{\floor{(nm-r-1)/n}+\floor{(nm-r-1)/m}} x_{r\bmod n}y_{r\bmod m} \\
= (-1)^{n+m+\floor{-(r+1)/n}+\floor{-(r+1)/m}} x_{r\bmod n}y_{r\bmod m}\mathrlap{.} \end{split}
\]
Using the fact that $\floor{-(r+1)/s}=-(\floor{r/s}+1)$ for positive integers $r$, $s$ and that $n+m$ is odd this becomes
\[ (-1)^{1+\floor{r/n}+\floor{r/m}} x_{r\bmod n}y_{r\bmod m} \]
which is the negative of the $r$th entry of $X\hattimes Y$ as required.
\end{proof}

\begin{example}
The palindromic odd perfect sequences $X=[\verb|++|]$
and $Y=[\verb|+q+|]$ used with Lemma~\ref{lem:oddprod}
produces the antipalindromic odd perfect sequence
$[\verb|+q-+Q-|]$.  (We use $\verb|Q|$ to denote $-q$.)
\end{example}

We now show that infinitely many palindromic odd perfect $Q_{8}$-sequences exist
using a variant of Theorem~\ref{thm:negcon} and the Barrera Acevedo--Dietrich correspondence.

\begin{theorem}\label{thm:oddperfect}
If $(A,B)$ is a Golay pair then
\[ P \coloneqq [-A; jB; k\tilde B; i\tilde A; iA; kB; j\tilde B; -\tilde A] \]
is a palindromic odd perfect $Q_8$-sequence.
\end{theorem}

\begin{proof}
This follows by a direct but tedious calculation of $\R_{P}(t)$; we give
an example of how this may be done.
Let $P=[P';\tilde P']$ have length~$8n$.  We find for $0\leq t\leq 4n$ that
\begin{flalign*}
\R_{P}(t) = C_{[P';\tilde P']}(t) - C_{[P';\tilde P']}(8n-t) ^* &&& \\
&&&\mathllap{= C_{P'}(t) + C_{\tilde P'}(t) + C_{\tilde P',P'}(4n-t)^* - C_{P',\tilde P'}(4n-t)^*.}
\end{flalign*}
Suppose that $1\leq t\leq n$.  Then $C_{P'}(t)=\phi(t)+\psi(n-t)^*$ where
\begin{align*}
\phi(t) &\coloneqq C_{-A}(t) + C_{jB}(t) + C_{k\tilde B}(t) + C_{i\tilde A}(t) \\
\psi(t) &\coloneqq C_{jB,-A}(t) + C_{k\tilde B,jB}(t) + C_{i\tilde A,k\tilde B}(t) .
\end{align*}
Since multiplying a sequence by a constant does not change its autocorrelation values
and since $(A,B,\tilde A,\tilde B)$ are complementary we find that $\phi(t)=0$.  Furthermore,
using the fact $C_{xA,yB}(t)=xy^*C_{A,B}(t)$ and $C_{\tilde A,\tilde B}(t)=C_{B,A}(t)$ (since $A$
and $B$ have real entries) we find
\[ \psi(t) = -jC_{B,A}(t) + i C_{\tilde B, B}(t) + j C_{B,A}(t) = i C_{\tilde B,B}(t) . \]
Additionally, we have $C_{\tilde P'}(t)=C_{P'}(t)^*$, $C_{P',\tilde P'}(4n-t)=C_{-A,-\tilde A}(n-t)=C_{A,\tilde A}(n-t)$, and $C_{\tilde P',P'}(4n-t)=C_{iA,i\tilde A}(n-t)=C_{A,\tilde A}(n-t)$.
Then
\[ \begin{split}
\R_{P}(t) = (i C_{\tilde B,B}(n-t))^* + i C_{\tilde B,B}(n-t) {\qquad\,} \\
{\qquad\,} + C_{A,\tilde A}(n-t)^* - C_{A,\tilde A}(n-t)^* = 0 .
\end{split} \]
Similarly one can show $\R_P(t)=0$ for $n<t\leq 4n$ from which the symmetry $\R_X(t)=\R_X(-t)^*$
implies $\R_P(t)=0$ for all $1\leq t<8n$.
\end{proof}

\begin{example}
Using the Golay pair $(\verb|++|,\verb|+-|)$ with Theorem~\ref{thm:oddperfect}
produces the palindromic odd perfect sequence $[\verb|--jJKkiiiikKJj--|]$. (We
denote $-i$ by $\verb|I|$, $-j$ by $\verb|J|$, and $-k$ by $\verb|K|$.)
\end{example}

In particular, palindromic odd perfect sequences exist in all lengths that are a power of two
since Theorem~\ref{thm:oddperfect} implies they exist in all lengths of the form
$2^a5^b13^c$ with $a\geq b+c+3$ and they exist in the lengths~$2$ and~$4$ as
shown by the examples $[\verb|++|]$ and $[\verb|+ii+|]$.
Using the fact that palindromic nega Williamson sequences exist in all odd lengths
up to~$33$ we used the Barrera~Acevedo--Dietrich correspondence to construct palindromic
odd perfect sequences in all odd lengths up to~$33$.
(This relies on the fact that a set of palindromic real sequences $\{A,B,C,D\}$
are necessarily ``nega-amicable'' in that $\R_{X,Y}(t)=\R_{Y,X}(t)$
for all~$t$ and $X,Y\in\{A,B,C,D\}$.)
Furthermore, using the fact that palindromic odd perfect $Q_8$\nobreakdash-sequences
exist in all lengths that are powers of two we used Lemma~\ref{lem:oddprod} to
construct palindromic odd perfect sequences in all even lengths up to~$68$
(see the appendix for an explicit list).
The Barrera Acevedo--Dietrich correspondence applied to these sequences gives
palindromic nega Williamson sequences in all even lengths up to~$68$.
It is conceivable that palindromic odd perfect sequences
and palindromic nega Williamson sequences actually exist
in all even lengths.

\section{Perfect quaternion sequences}\label{sec:perfect}

We now use our constructions for Williamson sequences and palindromic nega
Williamson sequences to show that Williamson sequences and
perfect $Q_8$-sequences exist in all lengths $2^t$
with $t\geq0$.

\begin{theorem}\label{thm:perfect}
Symmetric perfect sequences over $Q_8$ exist for all lengths $2^t$.
\end{theorem}

\begin{proof}
First, note that Golay sequences exist in all lengths $2^t$.
By Theorem~\ref{thm:negcon} and Lemma~\ref{lem:pal}
it follows that antipalindromic nega Williamson sequences with the $Q_8$-property
exist in all lengths~$2^t$ for $t\geq2$ (they also exist for smaller $t$,
see Example~\ref{ex:simple}).
Theorem~\ref{thm:main}
then implies that if Williamson sequences with the $Q_8$-property
exist in length $2^t$ they also exist in length $2^{t+2}$ for all $t\geq1$.
Additionally, it is known that Williamson sequences with the $Q_8$-property
exist in lengths~$2$ and~$4$ (see below).
By induction, Williamson sequences
with the $Q_8$-property exist in all lengths $2^t$ for $t\geq0$.
By the Barrera Acevedo--Dietrich correspondence symmetric perfect $Q_8$-sequences
exist in all lengths $2^t$ as well.
\end{proof}

\begin{example}\label{ex:perfect}
We use the base Golay pair $(\verb|+|,\verb|+|)$, the base Williamson sequences
$(\verb|+|,\verb|+|,\verb|+|,\verb|+|)$, $(\verb|+-|,\verb|+-|,\verb|++|,\verb|++|)$,
and the base nega Williamson sequences $(\verb|++|,\verb|++|,\verb|++|,\verb|++|)$.
Additionally, we use Golay's interleaving doubling construction~\cite{golay1961complementary}
to generate larger Golay pairs via the mapping $(A,B)\mapsto(A\shuffle B,A\shuffle -B)$.
From Theorem~\ref{thm:altmain} we generate the Williamson sequences
$(\verb|++-+|,\verb|++-+|,\verb|++-+|,\verb|++-+|)$, from Theorem~\ref{thm:negcon}
and Lemma~\ref{lem:pal} we generate the antipalindromic nega Williamson sequences
$(\verb|++++|,\verb|++++|,\verb|-++-|,\verb|-++-|)$, and from Theorem~\ref{thm:main}
we generate the Williamson sequences $(\verb|+-+++++-|,\verb|+-+++++-|,\verb|+--+++--|,\verb|+--+++--|)$.
Continuing in this way and using the Barrera Acevedo--Dietrich correspondence produces
perfect $Q_8$-sequences of lengths $2^t$ for all $t\geq0$.
We denote $-i$ by $\verb|I|$, $-j$ by $\verb|J|$, and $-k$ by $\verb|K|$
and explicitly give the sequences produced with this construction for $t\leq7$:
\[ [\verb|-|],\; [\verb|-j|],\; [\verb|--+-|],\; [\verb|-+j---j+|],\; [\verb|-+-J+j-----j+J-+|] \]
\[ [\verb|-i++jJ-K-k-jj-+I-I+-jj-k-K-Jj++i|], \]
\[ \mathllap{[}\verb|-+++-iJI+kjK-J-J-j-j-kjK+iJI--+-| \]\vspace{-1.6\baselineskip}
\[ \verb|--+--IJi+Kjk-j-j-J-J-Kjk+IJi-+++|\mathrlap{],} \]
\[ \mathllap{[}\verb|-iii+i+Ij+J+--K+-Jkj-JjJjk-K+KIK| \]\vspace{-1.6\baselineskip}
\[ \verb|-kIk+k-Kjjjj-Jkj--K+--J-ji+I+IiI| \]\vspace{-1.6\baselineskip}
\[ \verb|-IiI+I+ij-J--+K--jkJ-jjjjK-k+kIk| \]\vspace{-1.6\baselineskip}
\[ \verb|-KIK+K-kjJjJ-jkJ-+K--+J+jI+i+iii|\mathrlap{]} \]
\end{example}

The perfect sequences generated by Theorem~\ref{thm:perfect} for $t\geq7$ are counterexamples
to the quaternionic form of Mow's conjecture presented by Blake~\cite{blake2016constructions}
because they can be generated using an orthogonal matrix construction, as we now show.

\begin{theorem}\label{thm:matrices}
Let $n\geq32$ and let~$M$ be the $(n/4)\times 4$ matrix containing
the entries of a perfect sequence~$P$ generated by Theorem~\ref{thm:perfect}
using the second set from Theorem~\ref{thm:negcon}.
Write the entries of~$P$ in~$M$ from left to right and top to bottom (i.e.,
$M_{i,j}=P_{4i+j}$).
Then the columns of~$M$ have the array orthogonality property.
\end{theorem}

\begin{proof}
Let $M_0$, $M_1$, $M_2$, $M_3$ denote the columns of $M$ as sequences.  We must
show that $\sum_{r=0}^{3}R_{M_r}(t)=0$ for all $1\leq t<n/4$ and that
$R_{M_r,M_s}(t)=0$ for all $t$ and $0\leq r<s<4$.

By construction $P=(M_0\shuffle M_2)\shuffle(M_1\shuffle M_3)$ is a perfect sequence and therefore
\[ R_{(M_0\shuffle M_2)\shuffle(M_1\shuffle M_3)}(t)=0 \quad \text{for all $1\leq t<n$.} \]
By property~4 in Section~\ref{sec:construction} this implies
$R_{M_0\shuffle M_2}(t)+R_{M_1\shuffle M_3}(t)=0$ for all $1\leq t<n/2$
and $\sum_{r=0}^{3}R_{M_r}(t)=0$ for all $1\leq t<n/4$.

Since $P$ was generated by applying the Barrera Acevedo--Dietrich correspondence
to the sequences generated in Theorem~\ref{thm:main} we have
$P=\doub(A)\shuffle\n(B)$ for some perfect symmetric sequence $A$ and antipalindromic sequence $B$ of even length~$n/4$.
Let $X'$ denote the sequence formed by the even entries of $X$
and let $X''$ denote the sequence formed by the odd entries of~$X$.
Then we have
\[ M_0 = \doub(A') , \; M_1 = \n(B') , \; M_2 = \doub(A'') , \; M_3 = \n(B'') . \]
By property~3 of Section~\ref{sec:construction} this representation yields
$R_{M_0,M_1}(t)=R_{M_0,M_3}(t)=R_{M_1,M_2}(t)=R_{M_2,M_3}(t)=0$ for all $t$
and only the crosscorrelation of the pairs $(M_0,M_2)$ and $(M_1,M_3)$ are left to consider.
Since $A=A'\shuffle A''$ is perfect and was generated by applying Theorem~\ref{thm:main}
we have that $A'$ is of the form $\doub(C_1)$ for some $C_1$ and $A''$ is of the form $\n(C_2)$
for some $C_2$.  Property~3 of Section~\ref{sec:construction} then yields
$R_{M_0,M_2}(t)=2R_{A',A''}(t)=2R_{\doub(C_1),\n(C_2)}(t)=0$ for all $t$.

Lastly, we must show $R_{M_1,M_3}(t)=0$.
Suppose the antipalindromic $B$ was generated from Theorem~\ref{thm:negcon} using the Golay pair $(D,E)$.
An analysis of the Barrera Acevedo--Dietrich correspondence shows that we have
\[ B' = [iD; j\tilde E; D; -k\tilde E], \quad B'' = [kE; -\tilde D; -jE; -i\tilde D] . \]
We have $\R_{B',B''}(t) = C_{B',B''}(t) - C_{B'',B'}(n/8-t)^*$ by definition.
Suppose that $0\leq t\leq n/32$ so that $C_{B'',B'}(n/8-t)=C_{kE,-k\tilde E}(n/32-t)=-C_{E,\tilde E}(n/32-t)$
and $C_{B',B''}(t)=\phi(t)+\psi(n/32-t)^*$ where
\begin{align*}
&\begin{varwidth}{\columnwidth}$\mspace{1.5mu}\phi(t) \coloneqq C_{iD,kE}(t) + C_{j\tilde E,-\tilde D}(t) + C_{D,-jE}(t) + C_{-k\tilde E,-i\tilde D}(t),$\end{varwidth} \\
&\begin{varwidth}{\columnwidth}$\psi(t) \coloneqq C_{-\tilde D,iD}(t) + C_{-jE,j\tilde E}(t) + C_{-i\tilde D,D}(t).$\end{varwidth}
\end{align*}
Note that $C_{\tilde E,\tilde D}(t)=C_{D,E}(t)$ since $D$ and $E$ contain real entries.  Then
\begin{align*}
\phi(t) &= jC_{D,E}(t) - jC_{\tilde E,\tilde D}(t) + jC_{D,E}(t) -jC_{\tilde E,\tilde D}(t) = 0, \\
\psi(t) &= -iC_{\tilde D,D}(t) - C_{E,\tilde E}(t) + iC_{\tilde D,D}(t) = -C_{E,\tilde E}(t) .
\end{align*}
Finally, $R_{M_1,M_3}(t)=2\R_{B',B''}(t)=2\bigl(-C_{E,\tilde E}(n/32-t)^*-(-C_{E,\tilde E}(n/32-t))^*\bigr)=0$ for $0\leq t\leq n/32$.
A similar calculation shows $\R_{B',B''}(t)=0$ for $n/32<t<n/8$ from which it follows that $R_{M_1,M_3}(t)=0$
for all~$t$. 
\end{proof}

\begin{example}\label{ex:matrices}
For $n=16$ we use the perfect sequence found in Example~\ref{ex:perfect}
and the matrix it generates has the array orthogonality property as well.
Otherwise we give the matrices constructed using Theorem~\ref{thm:matrices}
for $n=32$, $64$, and $128$.  The transpose of the matrices are displayed
to save space.
\[\begin{gathered} \setlength\arraycolsep{0pt}
\begin{bmatrix}
\verb|-| & \verb|+| & \verb|-| & \verb|+| \\
\verb|+| & \verb|j| & \verb|-| & \verb|J| \\
\verb|-| & \verb|-| & \verb|-| & \verb|-| \\
\verb|J| & \verb|-| & \verb|j| & \verb|+| 
\end{bmatrix}\;
\begin{bmatrix}
\verb|-| & \verb|j| & \verb|-| & \verb|j| & \verb|-| & \verb|j| & \verb|-| & \verb|j| \\
\verb|I| & \verb|J| & \verb|-| & \verb|k| & \verb|i| & \verb|j| & \verb|+| & \verb|K| \\
\verb|+| & \verb|-| & \verb|-| & \verb|+| & \verb|+| & \verb|-| & \verb|-| & \verb|+| \\
\verb|K| & \verb|+| & \verb|j| & \verb|i| & \verb|k| & \verb|-| & \verb|J| & \verb|I| 
\end{bmatrix}\;
\begin{bmatrix}
\verb|-| & \verb|-| & \verb|+| & \verb|-| & \verb|-| & \verb|-| & \verb|+| & \verb|-| & \verb|-| & \verb|-| & \verb|+| & \verb|-| & \verb|-| & \verb|-| & \verb|+| & \verb|-| \\
\verb|I| & \verb|K| & \verb|j| & \verb|+| & \verb|-| & \verb|j| & \verb|K| & \verb|i| & \verb|i| & \verb|k| & \verb|J| & \verb|-| & \verb|+| & \verb|J| & \verb|k| & \verb|I| \\
\verb|+| & \verb|J| & \verb|j| & \verb|-| & \verb|-| & \verb|j| & \verb|J| & \verb|+| & \verb|+| & \verb|J| & \verb|j| & \verb|-| & \verb|-| & \verb|j| & \verb|J| & \verb|+| \\
\verb|I| & \verb|k| & \verb|J| & \verb|+| & \verb|-| & \verb|J| & \verb|k| & \verb|i| & \verb|i| & \verb|K| & \verb|j| & \verb|-| & \verb|+| & \verb|j| & \verb|K| & \verb|I| 
\end{bmatrix}\\ \setlength\arraycolsep{0pt}
\begin{bmatrix}
\verb|-| & \verb|+| & \verb|j| & \verb|-| & \verb|-| & \verb|-| & \verb|j| & \verb|+| & \verb|-| & \verb|+| & \verb|j| & \verb|-| & \verb|-| & \verb|-| & \verb|j| & \verb|+| & \verb|-| & \verb|+| & \verb|j| & \verb|-| & \verb|-| & \verb|-| & \verb|j| & \verb|+| & \verb|-| & \verb|+| & \verb|j| & \verb|-| & \verb|-| & \verb|-| & \verb|j| & \verb|+| \\
\verb|I| & \verb|I| & \verb|K| & \verb|K| & \verb|J| & \verb|j| & \verb|-| & \verb|+| & \verb|-| & \verb|-| & \verb|j| & \verb|j| & \verb|k| & \verb|K| & \verb|I| & \verb|i| & \verb|i| & \verb|i| & \verb|k| & \verb|k| & \verb|j| & \verb|J| & \verb|+| & \verb|-| & \verb|+| & \verb|+| & \verb|J| & \verb|J| & \verb|K| & \verb|k| & \verb|i| & \verb|I| \\
\verb|I| & \verb|K| & \verb|J| & \verb|+| & \verb|-| & \verb|j| & \verb|k| & \verb|i| & \verb|i| & \verb|k| & \verb|j| & \verb|-| & \verb|+| & \verb|J| & \verb|K| & \verb|I| & \verb|I| & \verb|K| & \verb|J| & \verb|+| & \verb|-| & \verb|j| & \verb|k| & \verb|i| & \verb|i| & \verb|k| & \verb|j| & \verb|-| & \verb|+| & \verb|J| & \verb|K| & \verb|I| \\
\verb|I| & \verb|i| & \verb|k| & \verb|K| & \verb|J| & \verb|J| & \verb|+| & \verb|+| & \verb|-| & \verb|+| & \verb|J| & \verb|j| & \verb|k| & \verb|k| & \verb|i| & \verb|i| & \verb|i| & \verb|I| & \verb|K| & \verb|k| & \verb|j| & \verb|j| & \verb|-| & \verb|-| & \verb|+| & \verb|-| & \verb|j| & \verb|J| & \verb|K| & \verb|K| & \verb|I| & \verb|I| 
\end{bmatrix}
\end{gathered}
\]
\end{example}

\section{Previous Work}\label{sec:previouswork}

Despite an enormous amount of work
the conjecture that the longest perfect sequences
over the complex $n$th roots of unity have length~$n^2$ remains open, though some
special cases of this conjecture have been resolved.
In the case $n=2$, Turyn~\cite{turyn1965character} showed that
any perfect binary sequence of length longer than four
must have length~$4m^2$ for odd~$m$ that is not a prime power.
Despite this progress the general conjecture 
even for $n=2$ remains open~\cite{leung2012new}.
In the case $n=4$, the conjecture states that the longest
perfect sequence over $\{\pm1,\pm i\}$
has length~$16$ and
Turyn~\cite{turyn1970complex} showed that the length of
any longer perfect sequence cannot be of the
form $2p^k$ for any prime $p$ and integer~$k$.
Most recently, Ma and Ng~\cite{ma2009non} showed many restrictions on the
length of perfect sequences over the $p$th roots of unity
for prime $p$.  In particular, they showed that no perfect
sequences of length $2p^{k+2}$ or $p^{k+3}$ exist for $k\geq0$.

Several other constructions for perfect sequences over complex roots of unity
have been found since the construction of Frank and Heimiller.
In 1972, Chu described a method~\cite{chu1972polyphase} of producing perfect sequences
of any length $n$.  Shortly after Chu's paper
was published, Frank published a response~\cite{frank1973comments} pointing out that
Zadoff had discovered the same construction and had been granted a patent
for a communication system based on his construction~\cite{zadoff1963phase}.
Other variants of Zadoff and Chu's sequences have been described by
Alltop~\cite{alltop1984decimations} and Lewis and Kretschmer~\cite{lewis1982linear}.
In 1983, Milewski~\cite{milewski1983periodic} found a new construction for perfect sequences
of length $n^{2t+1}$ over $n^{t+1}$th roots of unity for all $n,t\geq1$.
In 2004, Liu and Fan~\cite{liu2004modified} found a new construction for perfect sequences
of length~$n$ over $n$th roots of unity when~$n$ is a multiple of four.
In 2014, Blake and Tirkel~\cite{blake2014construction} gave a construction
for perfect sequences of length~$4mn^{t+1}$ over $2mn^t$th roots of unity for $m,n,t\geq1$.

Blake has run extensive searches for perfect sequences over
complex roots of unity and quaternions~\cite{blake2016constructions}
and has found a number of perfect sequences over the $n$th roots of unity
that cannot be generated using matrices with array orthogonality (like
those from the Frank--Heimiller and Mow constructions).
The sequences that he found are counterexamples to the conjecture
that Mow's unified construction produces all perfect sequences
over $n$th roots of unity~\cite{mow1996new}
but are shorter than~$n^2$ and thus
are not counterexamples to Mow's original conjecture~\cite{mow1995unified}
(also generalized by
L\"uke, Schotten, and Hadinejad-Mahram~\cite{luke2003binary}
to odd perfect sequences).

Apparently no infinite family of odd perfect $Q_8$-sequences
has been previously constructed.\footnote{A construction given in~\cite{blake2016constructions}
uses the term `odd-perfect' but with an alternative meaning that $R_A(t)=0$ for odd $t$.}
It is known that no infinite family of odd perfect sequences can exist
over the alphabet~$\{\pm1\}$---L\"uke, Schotten, and Hadinejad-Mahram show that the
longest odd perfect $\{\pm1\}$-sequence has length two
and they conjecture
that the longest odd perfect sequence over the alphabet $\{\pm1,\pm i\}$
has length four~\cite{luke2003binary}.
However, L\"uke has constructed \emph{almost} perfect
and odd perfect $\{\pm1,\pm i\}$\nobreakdash-sequences~$A$
(with $R_A(t)=0$ or $\R_A(t)=0$ for all $1\leq t<n$ except $t=n/2$) in all lengths $q+1$
where $q\equiv1\pmod4$ is a prime power~\cite{luke2001almost}.

Some work has also been done on perfect quaternion arrays.
In 2013, Barrera Acevedo and Jolly~\cite{acevedo2014perfect} constructed perfect $Q_8$-arrays
of size $2\times(p+1)/2$ for primes $p\equiv1\pmod{4}$.  Furthermore,
they extended a construction of Arasu and de Launey~\cite{arasu2001two}
to produce perfect $Q_8$-arrays of size $2p\times p(p+1)/2$ for primes $p\equiv1\pmod{4}$.
Additionally, Blake~\cite{blake2016constructions} found a construction for perfect $Q_8$-arrays of
size $2^t\times2^t$ with $2\leq t<7$.

There does not seem to be much prior work on nega Williamson sequences, though
Xia, Xia, Seberry, and Wu~\cite{xia2006hadamard}
study them under the name ``4-suitable negacyclic matrices''.
They give constructions for them
in terms of Golay pairs, Williamson sequences, and base sequences.
However, we could find no prior constructions that specifically generated
palindromic or antipalindromic nega Williamson sequences.
L\"uke presents a method~\cite{luke1997binary} of constructing pairs
of negacomplementary binary sequences,
also studied under the name ``associated pairs'' by Ito~\cite{ito2000hadamard}
and ``negaperiodic Golay pairs'' by Balonin and \DJ{}okovi\'c~\cite{balonin2015negaperiodic}.
L\"uke and Schotten~\cite{luke1995odd} also give a construction for
odd perfect almost binary sequences
that Wen, Hu, and Jin~\cite{wen2004design} use to construct negacomplementary binary sequences.
Jin et~al.~\cite{jin2009necessary} give necessary conditions for the existence of
negacomplementary sequences and a doubling construction for negacomplementary sequences.
Yang, Tang, and Zhou~\cite{yang2015binary} show that a $\{\pm1\}$-sequence~$A$ of length~$n$
satisfies $\max_{1\leq t<n}\bigl\lvert\R_A(t)\bigr\rvert-1\geq(n-1)\bmod{2}$
and give constructions for sequences that are optimal,
i.e., ones that meet the bound exactly.

Williamson sequences have been quite well-studied since introduced
by Williamson 75 years ago~\cite{williamson1944hadamard}.
They are often presented using matrix notation and known as ``Williamson matrices''
since one of their traditional applications has been to construct Hadamard matrices.
In this formulation Williamson matrices are defined to be circulant
(i.e., each row is the cyclic shift of the previous row)
and Williamson sequences are simply the first rows of Williamson matrices.
Baumert and Hall~\cite{baumert1965hadamard} performed an exhaustive search
for Williamson sequences of odd length $n\leq23$ and presented a doubling
construction for a generalization of Williamson matrices
that are symmetric but not circulant.  Turyn~\cite{turyn1972infinite}
constructed Williamson sequences in all lengths $(q+1)/2$ where
$q\equiv1\pmod{4}$ is a prime power, Whiteman~\cite{whiteman1976hadamard}
constructed Williamson sequences in all lengths $q(q+1)/2$ where $q\equiv1\pmod{4}$
is a prime power, and Spence~\cite{spence1977infinite} constructed Williamson sequences
in all lengths $q^t(q+1)/2$ where $q\equiv1\pmod{4}$ is a prime power and $t\geq0$.

Computer searches have determined that Williamson sequences in odd lengths are
particularly rare.
Following Baumert and Hall's exhaustive search,
Baumert~\cite{baumert1966hadamard} found a new set of Williamson sequences
in length~$29$, Sawade~\cite{40002775009}
found eight new sets in lengths~$25$ and~$27$,
Yamada~\cite{yamada1979williamson} found one new set in length~$37$,
Koukouvinos and Kounias~\cite{koukouvinos1988hadamard} found four new sets in length~$33$,
\DJ{}okovi\'c~\cite{dokovic1992williamson,dokovic1993williamson,dokovic1995note}
found six new sets in the lengths $25$, $31$, $33$, $37$, and $39$, van Vliet found one new set in length~$51$
(unpublished but appears in~\cite{holzmann2008williamson}),
Holzmann, Kharaghani, and Tayfeh-Rezaie~\cite{holzmann2008williamson} found one new set in length $43$,
and Bright, Kotsireas, and Ganesh~\cite{bright2019applying} found one new set in length $63$.
\DJ{}okovi\'c~\cite{dokovic1993williamson} also found that no Williamson sequences exist
in length $35$ and Holzmann, Kharaghani, and Tayfeh-Rezaie~\cite{holzmann2008williamson}
found that no Williamson sequences exist in lengths $47$, $53$, and~$59$.

In even lengths Williamson sequences are much more common, as first shown by an exhaustive
search up to length~$18$ by Kotsireas and Koukouvinos~\cite{kotsireas2006constructions}.
Non-exhaustive searches were performed up to length~$34$ by Bright et~al.~\cite{bright2016mathcheck2},
and up to length~$42$ by Zulkoski et~al.~\cite{zulkoski2017}.
Bright~\cite{brightthesis} completed an exhaustive search up to length~$44$ and
Bright, Kotsireas, and Ganesh~\cite{bright2019applying} completed an exhaustive search in
the even lengths up to~$70$.  Barrera Acevedo and Dietrich's correspondence~\cite{acevedo2019new}
can be used to generate Williamson sequences in many lengths including~$70$.

\section{Conclusion}\label{sec:conclusion}

We have shown that perfect and odd perfect $Q_8$-sequences exist in all lengths that are a power of two.
Barrera Acevedo and Dietrich~\cite{acevedo2018perfect} summarize the knowledge (as of 2018) of perfect $Q_8$-sequences
as follows:
\begin{quote}
Currently there exists only one infinite family of perfect sequences over
the quaternions (of magnitude one)\dots
\end{quote}
This infinite family was found by Barrera Acevedo and Hall~\cite{acevedo2012perfect}
who gave a construction for perfect $Q_8$-sequences
in lengths of the form $q+1$ where $q\equiv1\pmod{4}$ is a prime power.
Thus, our construction for lengths of the form~$2^t$ is the second known infinite family
of perfect $Q_8$-sequences.
Additionally, our construction for odd perfect $Q_8$-sequences is the first known infinite family
of odd perfect $Q_8$-sequences.
Because our perfect sequences can be constructed
using matrices with array orthogonality (as shown in Theorem~\ref{thm:matrices})
they disprove Blake's conjecture~\cite[Conjecture 8.2.1]{blake2016constructions}
that the longest perfect $Q_8$-sequences generated from an orthogonal array construction
have length $64$.  Theorem~\ref{thm:matrices} also implies the existence of perfect $Q_8$-arrays
of sizes $2^t\times 4$ for all~$t\geq2$ and the construction of
Barrera Acevedo and Jolly~\cite{acevedo2014perfect} then implies the existence
of perfect $Q_8$-arrays of size $2^tp\times 4p$ when $p=2^{t+2}-1$ is prime.

Furthermore, we generalize Williamson's doubling construction~\cite{williamson1944hadamard} from 1944,
showing that the existence of Williamson sequences of odd length $n$ implies not only
the existence of Williamson sequences of length $2n$ but also implies the existence
of Williamson sequences of length $2^tn$ for all $t\geq1$.
We have also shown the importance of nega Williamson sequences, a class of sequences
defined by Xia et~al.~\cite{xia2006hadamard} in 2006.  In particular, we have demonstrated
the importance of palindromic nega Williamson sequences.

Lastly, our constructions provide an explanation for the abundance of Williamson sequences
in lengths that are divisible by a large power of two.  Prior to this work
it was noticed that Williamson sequences are much
more abundant in these lengths.  For example, fewer than 100 sets of Williamson
sequences are known to exist in the odd lengths, but
an exhaustive computer search~\cite{bright2019applying} found 130{,}739
sets of Williamson sequences in the even lengths up to~$70$.
We found that it was possible to generate 95{,}759 (about 75\%) of these sets
using Theorems~\ref{thm:main} and~\ref{thm:altmain}.
Thus, these theorems provide an explanation for the existence of many Williamson sequences.
However, they still do not explain the existence of Williamson sequences in all even
lengths.  In particular, they only work in lengths that are multiples of four.

It would also be interesting to find a construction that works in the even lengths
that are not multiples of four.
Williamson's doubling result can be used for lengths $2n$ but requires
that Williamson sequences of odd length~$n$ exist.
Barrera Acevedo and Dietrich's correspondence can be used in certain cases even if Williamson sequences
of length~$n$ do not exist, assuming~$n$ is composite.
For example, the Barrera Acevedo--Dietrich correspondence
implies that Williamson sequences of length~$70$ exist since Williamson sequences
of length~$7$ exist and Williamson sequences with the $Q_8$-property of length~$10$ exist.
However, we could only generate about 40\%
of the Williamson sequences in length~$70$ using the Barrera Acevedo--Dietrich correspondence
suggesting that there is another construction for Williamson sequences
that is currently unknown.

\bibliographystyle{IEEEtran}
\bibliography{ieee}

\section*{Appendix}

\subsection{Proofs}

We now give proofs of the properties from Section~\ref{sec:construction}.
In each case we let $X=[x_0,\dotsc,x_{n-1}]$ and $Y=[y_0,\dotsc,y_{n-1}]$ be arbitrary
sequences of length $n$.

\begin{enumerate}
\item $R_{\doub(X)}(t)=2R_X(t)$.

The $r$th entry of $\doub(X)$ is $x_{r\bmod n}$.  Then
\[ R_{\doub(X)}(t)=\sum_{r=0}^{2n-1}x_{r\bmod n}x^*_{r+t\bmod n} = 2R_X(t) . \]
\item $R_{\n(X)}(t)=2\R_X(t)$.

The $r$th entry of $\n(X)$ is $(-1)^{\floor{r/n}}x_{r\bmod n}$.  Then $R_{\n(X)}(t)$ is
\begin{align*}
&\sum_{r=0}^{2n-1}(-1)^{\floor{r/n}}x_{r\bmod n}(-1)^{\floor{(r+t)/n}}x^*_{r+t\bmod n} \\
&= \sum_{r=0}^{n-1}x_r(-1)^{\floor{(r+t)/n}}x^*_{r+t\bmod n} \\
&\quad+ \sum_{r=0}^{n-1}(-x_r)(-1)^{\floor{(r+n+t)/n}}x^*_{r+t\bmod n} \\
&= 2\R_X(t) .
\end{align*}
\item $R_{\doub(X),\n(Y)}(t)=0$ and $R_{\n(Y),\doub(X)}(t)=0$.

The $r$th entry of $\doub(X)$ is $x_{r\bmod n}$ and the $r$th entry of $\n(Y)$ is $(-1)^{\floor{r/n}}y_{r\bmod n}$.  Then $R_{\doub(X),\n(Y)}(t)$ is
\begin{align*}
&\sum_{r=0}^{2n-1}x_{r\bmod n}(-1)^{\floor{(r+t)/n}}y^*_{r+t\bmod n} \\
&= \sum_{r=0}^{n-1}x_{r}(-1)^{\floor{(r+t)/n}}y^*_{r+t\bmod n} \\
&\quad+ \sum_{r=0}^{n-1}x_{r}(-1)^{\floor{(r+n+t)/n}}y^*_{r+t\bmod n} \\
&= \R_{X,Y}(t) - \R_{X,Y}(t) = 0.
\end{align*}
The second property follows because $R_{X,Y}(t)=R_{Y,X}(-t)^*$ for all $X$, $Y$, and $t$.
\item $R_{X\shuffle Y}(2t)=R_X(t)+R_Y(t)$.

The $(2r)$th entry of $X\shuffle Y$ is $x_r$ and the $(2r+1)$th entry is $y_r$.  Then $R_{X\shuffle Y}(2t)$ is
\begin{align*} &\sum_{\substack{r=0\\\text{$r$ even}}}^{2n-1}x_{r/2}x^*_{r/2+t\bmod n} + \sum_{\substack{r=0\\\text{$r$ odd}}}^{2n-1}y_{(r-1)/2}y^*_{(r-1)/2+t\bmod n} \\
&= R_X(t) + R_Y(t) . \end{align*}
\item $R_{X\shuffle Y}(2t+1)=R_{X,Y}(t)+R_{Y,X}(t+1)$.

The $(2r)$th entry of $X\shuffle Y$ is $x_r$ and the $(2r+1)$th entry is $y_r$.  Then $R_{X\shuffle Y}(2t+1)$ is
\begin{align*} &\sum_{\substack{r=0\\\text{$r$ even}}}^{2n-1}x_{r/2}y^*_{r/2+t\bmod n} + \sum_{\substack{r=0\\\text{$r$ odd}}}^{2n-1}y_{(r-1)/2}x^*_{(r+1)/2+t\bmod n} \\
&= R_{X,Y}(t) + R_{Y,X}(t+1) . \end{align*}
\item If $X$ is symmetric then $\doub(X)$ is symmetric.

Note that $x_{2n-r\bmod n}=x_{n-r\bmod n}=x_{r\bmod n}$.  Thus the $(2n-r)$th entry of $\doub(X)$
is equal to the $r$th entry, as required.
\item If $X$ is antipalindromic and of even length then $\n(X)$ is palindromic.

Let $Y$ be the first half of $X$ (i.e., $X=[Y;-\tilde Y]$) so that
$\n(X)=[Y;-\tilde Y;-Y;\tilde Y]$
is a palindrome.
\item If $X$ is symmetric and $Y$ is palindromic then $X\shuffle Y$ is symmetric.

First, we show the even entries of $X\shuffle Y$ satisfy the symmetric property.  The $(2r)$th entry of $X\shuffle Y$ is $x_r$
and the $(2n-2r)$th entry of $X\shuffle Y$ is $x_{n-r}$ (for $r\neq0$).  Since $X$ is symmetric $x_{n-r}=x_r$ showing the
$(2r)$th and $(2n-2r)$th entries are equal.

Second, we show the odd entries of $X\shuffle Y$ satisfy the symmetric property.
The $(2r+1)$th entry of $X\shuffle Y$ is $y_r$ and the $(2n-2r-1)$the entry of $X\shuffle Y$ is $y_{n-r-1}$.
Since $Y$ is palindromic $y_{n-r-1}=y_r$ showing the
$(2r+1)$th and $(2n-2r-1)$th entries are equal.

\item If $X$ is antisymmetric and of odd length then $\n(X)$ is symmetric.

Let $Y$ be the first half of the palindromic part of $X$ (i.e., $X=[x_0;Y;-\tilde Y]$)
so that
$\n(X) = [x_0;Y;-\tilde Y;-x_0;-Y;\tilde Y]$
is symmetric.

\item If $X$ is palindromic then $\doub(X)$ is palindromic.

Note that $x_{2n-r-1\bmod n}=x_{n-r-1\bmod n}=x_{r\bmod n}$.  Thus the $(2n-r-1)$th entry of $\doub(X)$
is equal to the $r$th entry, as required.
\end{enumerate}

\subsection{List of odd perfect quaternion sequences}

We now give palindromic odd perfect $\Q$-sequences in all lengths $n<70$ except for $35$, $47$, $53$, $59$, $65$, and $67$.
The sequences in odd lengths were constructed using Lemma~\ref{lem:negacor} with a
previously known set of Williamson sequences of length~$n$.  The sequences in even lengths were constructed using
Lemma~\ref{lem:oddprod}, Theorem~\ref{thm:oddperfect}, and Lemma~\ref{lem:pal} and are new
to the best of our knowledge, though perfect quaternion sequences
may be constructed in these lengths using the results of Barrera Acevedo and Dietrich~\cite{acevedo2019new}.

The sequences are denoted by $P_n$ where~$n$ is the length of the sequence.
The symbols~$\verb|+|$ and~$\verb|-|$ denote~$1$ and~$-1$,
capitalization denotes negation of an entry, and an overlined entry denotes
left multiplication by $q$, i.e., the symbol $\mathrlap{\overline{\color{white}{\phantom{\mathtt{I}}}}}\verb|I|$ denotes the entry $-qi$.

{\small\setlength{\parindent}{0pt}
$P_{1}=[\verb|+|]$

$P_{2}=[\verb|+|\verb|+|]$

$P_{3}=[\verb|I|\d\verb|-|\verb|I|]$

$P_{4}=[\verb|+|\verb|i|\verb|i|\verb|+|]$

$P_{5}=[\verb|J|\verb|+|\d\verb|+|\verb|+|\verb|J|]$

$P_{6}=[\verb|i|\d\verb|-|\verb|I|\verb|I|\d\verb|-|\verb|i|]$

$P_{7}=[\verb|j|\verb|-|\verb|J|\d\verb|-|\verb|J|\verb|-|\verb|j|]$

$P_{8}=[\verb|-|\verb|j|\verb|k|\verb|i|\verb|i|\verb|k|\verb|j|\verb|-|]$

$P_{9}=[\verb|K|\verb|j|\verb|I|\verb|+|\d\verb|+|\verb|+|\verb|I|\verb|j|\verb|K|]$

$P_{10}=[\verb|J|\verb|-|\d\verb|-|\verb|+|\verb|J|\verb|J|\verb|+|\d\verb|-|\verb|-|\verb|J|]$

$P_{11}=[\verb|i|\verb|+|\verb|K|\verb|-|\verb|J|\d\verb|-|\verb|J|\verb|-|\verb|K|\verb|+|\verb|i|]$

$P_{12}=[\verb|+|\d\verb|-|\verb|i|\verb|+|\d\verb|K|\verb|I|\verb|I|\d\verb|K|\verb|+|\verb|i|\d\verb|-|\verb|+|]$

$P_{13}=[\verb|K|\verb|i|\verb|J|\verb|K|\verb|j|\verb|i|\d\verb|+|\verb|i|\verb|j|\verb|K|\verb|J|\verb|i|\verb|K|]$

$P_{14}=[\verb|J|\verb|-|\verb|J|\d\verb|+|\verb|j|\verb|-|\verb|j|\verb|j|\verb|-|\verb|j|\d\verb|+|\verb|J|\verb|-|\verb|J|]$

$P_{15}=[\verb|J|\verb|j|\verb|-|\verb|J|\verb|J|\verb|+|\verb|+|\d\verb|-|\verb|+|\verb|+|\verb|J|\verb|J|\verb|-|\verb|j|\verb|J|]$

$P_{16}=[\verb|-|\verb|-|\verb|j|\verb|J|\verb|K|\verb|k|\verb|i|\verb|i|\verb|i|\verb|i|\verb|k|\verb|K|\verb|J|\verb|j|\verb|-|\verb|-|]$

$P_{17}=[\verb|+|\verb|i|\verb|K|\verb|+|\verb|+|\verb|j|\verb|j|\verb|J|\d\verb|+|\verb|J|\verb|j|\verb|j|\verb|+|\verb|+|\verb|K|\verb|i|\verb|+|]$

$P_{18}=[\verb|K|\verb|J|\verb|i|\verb|+|\d\verb|+|\verb|-|\verb|i|\verb|j|\verb|K|\verb|K|\verb|j|\verb|i|\verb|-|\d\verb|+|\verb|+|\verb|i|\verb|J|\verb|K|]$

$P_{19}=[\verb|k|\verb|K|\verb|j|\verb|-|\verb|J|\verb|J|\verb|K|\verb|+|\verb|+|\d\verb|-|\verb|+|\verb|+|\verb|K|\verb|J|\verb|J|\verb|-|\verb|j|\verb|K|\verb|k|]$

$P_{20}=[\verb|k|\verb|-|\d\verb|+|\verb|i|\verb|K|\verb|j|\verb|+|\d\verb|k|\verb|i|\verb|J|\verb|J|\verb|i|\d\verb|k|\verb|+|\verb|j|\verb|K|\verb|i|\d\verb|+|\verb|-|\verb|k|]$

$P_{21}=[\verb|+|\verb|+|\verb|I|\verb|I|\verb|i|\verb|+|\verb|+|\verb|i|\verb|+|\verb|-|\d\verb|+|\verb|-|\verb|+|\verb|i|\verb|+|\verb|+|\verb|i|\verb|I|\verb|I|\verb|+|\verb|+|]$

$P_{22}=[\verb|I|\verb|+|\verb|K|\verb|+|\verb|j|\d\verb|-|\verb|J|\verb|+|\verb|k|\verb|+|\verb|i|\verb|i|\verb|+|\verb|k|\verb|+|\verb|J|\d\verb|-|\verb|j|\verb|+|\verb|K|\verb|+|\verb|I|]$

$P_{23}=[\verb|j|\verb|+|\verb|k|\verb|-|\verb|i|\verb|-|\verb|-|\verb|K|\verb|J|\verb|i|\verb|I|\d\verb|-|\verb|I|\verb|i|\verb|J|\verb|K|\verb|-|\verb|-|\verb|i|\verb|-|\verb|k|\verb|+|\verb|j|]$

$P_{24}=[\verb|+|\d\verb|J|\verb|k|\verb|I|\d\verb|+|\verb|k|\verb|j|\d\verb|k|\verb|-|\verb|J|\d\verb|I|\verb|i|\verb|i|\d\verb|I|\verb|J|\verb|-|\d\verb|k|\verb|j|\verb|k|\d\verb|+|\verb|I|\verb|k|\d\verb|J|\verb|+|]$

$P_{25}=[\verb|k|\verb|J|\verb|J|\verb|I|\verb|i|\verb|k|\verb|k|\verb|+|\verb|j|\verb|-|\verb|+|\verb|i|\d\verb|+|\verb|i|\verb|+|\verb|-|\verb|j|\verb|+|\verb|k|\verb|k|\verb|i|\verb|I|\verb|J|\verb|J|\verb|k|]$

$P_{26}=[\verb|K|\verb|I|\verb|j|\verb|K|\verb|j|\verb|I|\d\verb|-|\verb|i|\verb|j|\verb|k|\verb|j|\verb|i|\verb|K|\verb|K|\verb|i|\verb|j|\verb|k|\verb|j|\verb|i|\d\verb|-|\verb|I|\verb|j|\verb|K|\verb|j|\verb|I|\verb|K|]$

$P_{27}=[\verb|I|\verb|i|\verb|k|\verb|J|\verb|I|\verb|I|\verb|-|\verb|j|\verb|+|\verb|K|\verb|-|\verb|J|\verb|K|\d\verb|-|\verb|K|\verb|J|\verb|-|\verb|K|\verb|+|\verb|j|\verb|-|\verb|I|\verb|I|\verb|J|\verb|k|\verb|i|\verb|I|]$

$P_{28}=[\verb|k|\verb|-|\verb|j|\d\verb|k|\verb|k|\verb|+|\verb|j|\verb|K|\verb|i|\verb|j|\d\verb|-|\verb|K|\verb|I|\verb|j|\verb|j|\verb|I|\verb|K|\d\verb|-|\verb|j|\verb|i|\verb|K|\verb|j|\verb|+|\verb|k|\d\verb|k|\verb|j|\verb|-|\verb|k|]$

$P_{29}=[\verb|+|\verb|i|\verb|I|\verb|I|\verb|K|\verb|k|\verb|j|\verb|J|\verb|k|\verb|I|\verb|I|\verb|K|\verb|K|\verb|K|\d\verb|+|\verb|K|\verb|K|\verb|K|\verb|I|\verb|I|\verb|k|\verb|J|\verb|j|\verb|k|\verb|K|\verb|I|\verb|I|\verb|i|\verb|+|]$

$P_{30}=[\verb|j|\verb|j|\verb|-|\verb|j|\verb|j|\verb|+|\verb|+|\d\verb|+|\verb|-|\verb|+|\verb|J|\verb|j|\verb|+|\verb|j|\verb|J|\verb|J|\verb|j|\verb|+|\verb|j|\verb|J|\verb|+|\verb|-|\d\verb|+|\verb|+|\verb|+|\verb|j|\verb|j|\verb|-|\verb|j|\verb|j|]$

$P_{31}=[\verb|-|\verb|I|\verb|I|\verb|+|\verb|i|\verb|-|\verb|I|\verb|-|\verb|-|\verb|-|\verb|I|\verb|i|\verb|I|\verb|i|\verb|i|\d\verb|-|\verb|i|\verb|i|\verb|I|\verb|i|\verb|I|\verb|-|\verb|-|\verb|-|\verb|I|\verb|-|\verb|i|\verb|+|\verb|I|\verb|I|\verb|-|]$

$P_{32}=[\verb|-|\verb|-|\verb|-|\verb|+|\verb|j|\verb|J|\verb|j|\verb|j|\verb|k|\verb|k|\verb|K|\verb|k|\verb|I|\verb|i|\verb|i|\verb|i|\verb|i|\verb|i|\verb|i|\verb|I|\verb|k|\verb|K|\verb|k|\verb|k|\verb|j|\verb|j|\verb|J|\verb|j|\verb|+|\verb|-|\verb|-|\verb|-|]$

$P_{33}=[\verb|J|\verb|+|\verb|I|\verb|j|\verb|K|\verb|K|\verb|k|\verb|j|\verb|k|\verb|k|\verb|+|\verb|I|\verb|i|\verb|k|\verb|-|\verb|i|\d\verb|+|\verb|i|\verb|-|\verb|k|\verb|i|\verb|I|\verb|+|\verb|k|\verb|k|\verb|j|\verb|k|\verb|K|\verb|K|\verb|j|\verb|I|\verb|+|\verb|J|]$

$P_{34}=[\verb|+|\verb|I|\verb|k|\verb|+|\verb|+|\verb|J|\verb|J|\verb|J|\d\verb|+|\verb|j|\verb|J|\verb|j|\verb|+|\verb|-|\verb|k|\verb|i|\verb|+|\verb|+|\verb|i|\verb|k|\verb|-|\verb|+|\verb|j|\verb|J|\verb|j|\d\verb|+|\verb|J|\verb|J|\verb|J|\verb|+|\verb|+|\verb|k|\verb|I|\verb|+|]$

$P_{36}=[\verb|J|\verb|J|\verb|I|\verb|i|\d\verb|k|\verb|+|\verb|i|\verb|K|\verb|J|\verb|K|\verb|J|\verb|-|\verb|I|\d\verb|-|\verb|+|\verb|+|\verb|k|\verb|K|\verb|K|\verb|k|\verb|+|\verb|+|\d\verb|-|\verb|I|\verb|-|\verb|J|\verb|K|\verb|J|\verb|K|\verb|i|\verb|+|\d\verb|k|\verb|i|\verb|I|\verb|J|\verb|J|]$

$P_{37}=[\verb|j|\verb|i|\verb|k|\verb|I|\verb|I|\verb|J|\verb|I|\verb|+|\verb|K|\verb|k|\verb|J|\verb|-|\verb|J|\verb|J|\verb|+|\verb|-|\verb|K|\verb|+|\d\verb|+|\verb|+|\verb|K|\verb|-|\verb|+|\verb|J|\verb|J|\verb|-|\verb|J|\verb|k|\verb|K|\verb|+|\verb|I|\verb|J|\verb|I|\verb|I|\verb|k|\verb|i|\verb|j|]$

$P_{38}=[\verb|K|\verb|K|\verb|j|\verb|+|\verb|j|\verb|J|\verb|K|\verb|-|\verb|-|\d\verb|-|\verb|+|\verb|-|\verb|k|\verb|J|\verb|J|\verb|+|\verb|J|\verb|K|\verb|k|\verb|k|\verb|K|\verb|J|\verb|+|\verb|J|\verb|J|\verb|k|\verb|-|\verb|+|\d\verb|-|\verb|-|\verb|-|\verb|K|\verb|J|\verb|j|\verb|+|\verb|j|\verb|K|\verb|K|]$

$P_{39}=[\verb|J|\verb|k|\verb|k|\verb|-|\verb|J|\verb|+|\verb|I|\verb|+|\verb|K|\verb|k|\verb|J|\verb|+|\verb|+|\verb|+|\verb|k|\verb|j|\verb|K|\verb|J|\verb|j|\d\verb|-|\verb|j|\verb|J|\verb|K|\verb|j|\verb|k|\verb|+|\verb|+|\verb|+|\verb|J|\verb|k|\verb|K|\verb|+|\verb|I|\verb|+|\verb|J|\verb|-|\verb|k|\verb|k|\verb|J|]$

$P_{40}=[\verb|k|\verb|K|\d\verb|I|\verb|+|\verb|j|\verb|-|\verb|K|\d\verb|K|\verb|I|\verb|I|\verb|+|\verb|-|\d\verb|+|\verb|J|\verb|I|\verb|K|\verb|i|\d\verb|j|\verb|j|\verb|j|\verb|j|\verb|j|\d\verb|j|\verb|i|\verb|K|\verb|I|\verb|J|\d\verb|+|\verb|-|\verb|+|\verb|I|\verb|I|\d\verb|K|\verb|K|\verb|-|\verb|j|\verb|+|\d\verb|I|\verb|K|\verb|k|]$

$P_{41}=[\verb|i|\verb|K|\verb|k|\verb|K|\verb|i|\verb|k|\verb|I|\verb|I|\verb|K|\verb|K|\verb|i|\verb|k|\verb|i|\verb|K|\verb|i|\verb|i|\verb|I|\verb|i|\verb|k|\verb|k|\d\verb|+|\verb|k|\verb|k|\verb|i|\verb|I|\verb|i|\verb|i|\verb|K|\verb|i|\verb|k|\verb|i|\verb|K|\verb|K|\verb|I|\verb|I|\verb|k|\verb|i|\verb|K|\verb|k|\verb|K|\\\phantom{P_{00}=[}\verb|i|]$

$P_{42}=[\verb|+|\verb|-|\verb|i|\verb|I|\verb|i|\verb|-|\verb|-|\verb|i|\verb|+|\verb|+|\d\verb|-|\verb|-|\verb|+|\verb|I|\verb|-|\verb|+|\verb|i|\verb|i|\verb|i|\verb|+|\verb|+|\verb|+|\verb|+|\verb|i|\verb|i|\verb|i|\verb|+|\verb|-|\verb|I|\verb|+|\verb|-|\d\verb|-|\verb|+|\verb|+|\verb|i|\verb|-|\verb|-|\verb|i|\verb|I|\verb|i|\\\phantom{P_{00}=[}\verb|-|\verb|+|]$

$P_{43}=[\verb|j|\verb|-|\verb|I|\verb|J|\verb|J|\verb|j|\verb|I|\verb|K|\verb|K|\verb|K|\verb|+|\verb|j|\verb|+|\verb|k|\verb|-|\verb|+|\verb|K|\verb|i|\verb|j|\verb|K|\verb|-|\d\verb|+|\verb|-|\verb|K|\verb|j|\verb|i|\verb|K|\verb|+|\verb|-|\verb|k|\verb|+|\verb|j|\verb|+|\verb|K|\verb|K|\verb|K|\verb|I|\verb|j|\verb|J|\verb|J|\\\phantom{P_{00}=[}\verb|I|\verb|-|\verb|j|]$

$P_{44}=[\verb|-|\verb|+|\verb|k|\verb|i|\verb|k|\d\verb|+|\verb|J|\verb|I|\verb|j|\verb|+|\verb|I|\verb|-|\verb|i|\verb|K|\verb|+|\verb|k|\d\verb|k|\verb|j|\verb|-|\verb|j|\verb|i|\verb|i|\verb|i|\verb|i|\verb|j|\verb|-|\verb|j|\d\verb|k|\verb|k|\verb|+|\verb|K|\verb|i|\verb|-|\verb|I|\verb|+|\verb|j|\verb|I|\verb|J|\d\verb|+|\verb|k|\\\phantom{P_{00}=[}\verb|i|\verb|k|\verb|+|\verb|-|]$

$P_{45}=[\verb|k|\verb|k|\verb|k|\verb|K|\verb|k|\verb|I|\verb|i|\verb|K|\verb|I|\verb|k|\verb|i|\verb|k|\verb|K|\verb|K|\verb|I|\verb|K|\verb|I|\verb|i|\verb|K|\verb|K|\verb|i|\verb|i|\d\verb|+|\verb|i|\verb|i|\verb|K|\verb|K|\verb|i|\verb|I|\verb|K|\verb|I|\verb|K|\verb|K|\verb|k|\verb|i|\verb|k|\verb|I|\verb|K|\verb|i|\verb|I|\\\phantom{P_{00}=[}\verb|k|\verb|K|\verb|k|\verb|k|\verb|k|]$

$P_{46}=[\verb|J|\verb|+|\verb|k|\verb|+|\verb|I|\verb|-|\verb|-|\verb|k|\verb|j|\verb|i|\verb|I|\d\verb|+|\verb|i|\verb|i|\verb|J|\verb|k|\verb|+|\verb|-|\verb|i|\verb|+|\verb|K|\verb|+|\verb|j|\verb|j|\verb|+|\verb|K|\verb|+|\verb|i|\verb|-|\verb|+|\verb|k|\verb|J|\verb|i|\verb|i|\d\verb|+|\verb|I|\verb|i|\verb|j|\verb|k|\verb|-|\\\phantom{P_{00}=[}\verb|-|\verb|I|\verb|+|\verb|k|\verb|+|\verb|J|]$

$P_{48}=[\verb|+|\d\verb|K|\verb|J|\verb|J|\d\verb|I|\verb|K|\verb|i|\d\verb|+|\verb|I|\verb|i|\d\verb|I|\verb|K|\verb|J|\d\verb|j|\verb|-|\verb|+|\d\verb|K|\verb|+|\verb|j|\d\verb|J|\verb|k|\verb|k|\d\verb|+|\verb|i|\verb|i|\d\verb|+|\verb|k|\verb|k|\d\verb|J|\verb|j|\verb|+|\d\verb|K|\verb|+|\verb|-|\d\verb|j|\verb|J|\verb|K|\d\verb|I|\verb|i|\verb|I|\\\phantom{P_{00}=[}\d\verb|+|\verb|i|\verb|K|\d\verb|I|\verb|J|\verb|J|\d\verb|K|\verb|+|]$

$P_{49}=[\verb|k|\verb|k|\verb|i|\verb|I|\verb|I|\verb|k|\verb|I|\verb|I|\verb|i|\verb|I|\verb|I|\verb|I|\verb|k|\verb|i|\verb|K|\verb|k|\verb|i|\verb|k|\verb|k|\verb|I|\verb|I|\verb|K|\verb|k|\verb|K|\d\verb|+|\verb|K|\verb|k|\verb|K|\verb|I|\verb|I|\verb|k|\verb|k|\verb|i|\verb|k|\verb|K|\verb|i|\verb|k|\verb|I|\verb|I|\verb|I|\\\phantom{P_{00}=[}\verb|i|\verb|I|\verb|I|\verb|k|\verb|I|\verb|I|\verb|i|\verb|k|\verb|k|]$

$P_{50}=[\verb|k|\verb|j|\verb|j|\verb|I|\verb|i|\verb|K|\verb|K|\verb|+|\verb|j|\verb|+|\verb|-|\verb|i|\d\verb|+|\verb|I|\verb|-|\verb|-|\verb|j|\verb|-|\verb|K|\verb|k|\verb|i|\verb|i|\verb|j|\verb|J|\verb|k|\verb|k|\verb|J|\verb|j|\verb|i|\verb|i|\verb|k|\verb|K|\verb|-|\verb|j|\verb|-|\verb|-|\verb|I|\d\verb|+|\verb|i|\verb|-|\\\phantom{P_{00}=[}\verb|+|\verb|j|\verb|+|\verb|K|\verb|K|\verb|i|\verb|I|\verb|j|\verb|j|\verb|k|]$

$P_{51}=[\verb|J|\verb|J|\verb|-|\verb|j|\verb|+|\verb|-|\verb|+|\verb|-|\verb|j|\verb|+|\verb|+|\verb|j|\verb|+|\verb|J|\verb|+|\verb|+|\verb|+|\verb|j|\verb|-|\verb|j|\verb|J|\verb|-|\verb|-|\verb|+|\verb|j|\d\verb|+|\verb|j|\verb|+|\verb|-|\verb|-|\verb|J|\verb|j|\verb|-|\verb|j|\verb|+|\verb|+|\verb|+|\verb|J|\verb|+|\verb|j|\\\phantom{P_{00}=[}\verb|+|\verb|+|\verb|j|\verb|-|\verb|+|\verb|-|\verb|+|\verb|j|\verb|-|\verb|J|\verb|J|]$

$P_{52}=[\verb|J|\verb|I|\verb|J|\verb|j|\verb|k|\verb|i|\d\verb|-|\verb|+|\verb|K|\verb|k|\verb|J|\verb|-|\verb|j|\verb|k|\verb|i|\verb|K|\verb|j|\verb|j|\verb|I|\d\verb|K|\verb|+|\verb|J|\verb|K|\verb|K|\verb|-|\verb|K|\verb|K|\verb|-|\verb|K|\verb|K|\verb|J|\verb|+|\d\verb|K|\verb|I|\verb|j|\verb|j|\verb|K|\verb|i|\verb|k|\verb|j|\\\phantom{P_{00}=[}\verb|-|\verb|J|\verb|k|\verb|K|\verb|+|\d\verb|-|\verb|i|\verb|k|\verb|j|\verb|J|\verb|I|\verb|J|]$

$P_{54}=[\verb|i|\verb|i|\verb|k|\verb|j|\verb|i|\verb|I|\verb|-|\verb|J|\verb|-|\verb|K|\verb|-|\verb|j|\verb|k|\d\verb|-|\verb|K|\verb|j|\verb|+|\verb|K|\verb|+|\verb|J|\verb|+|\verb|I|\verb|I|\verb|j|\verb|K|\verb|i|\verb|I|\verb|I|\verb|i|\verb|K|\verb|j|\verb|I|\verb|I|\verb|+|\verb|J|\verb|+|\verb|K|\verb|+|\verb|j|\verb|K|\\\phantom{P_{00}=[}\d\verb|-|\verb|k|\verb|j|\verb|-|\verb|K|\verb|-|\verb|J|\verb|-|\verb|I|\verb|i|\verb|j|\verb|k|\verb|i|\verb|i|]$

$P_{55}=[\verb|J|\verb|J|\verb|J|\verb|+|\verb|+|\verb|j|\verb|J|\verb|+|\verb|+|\verb|-|\verb|+|\verb|-|\verb|+|\verb|+|\verb|j|\verb|j|\verb|-|\verb|j|\verb|+|\verb|-|\verb|J|\verb|+|\verb|-|\verb|-|\verb|j|\verb|+|\verb|j|\d\verb|+|\verb|j|\verb|+|\verb|j|\verb|-|\verb|-|\verb|+|\verb|J|\verb|-|\verb|+|\verb|j|\verb|-|\verb|j|\\\phantom{P_{00}=[}\verb|j|\verb|+|\verb|+|\verb|-|\verb|+|\verb|-|\verb|+|\verb|+|\verb|J|\verb|j|\verb|+|\verb|+|\verb|J|\verb|J|\verb|J|]$

$P_{56}=[\verb|k|\verb|K|\verb|+|\d\verb|+|\verb|J|\verb|j|\verb|i|\verb|k|\verb|I|\verb|i|\d\verb|I|\verb|j|\verb|-|\verb|+|\verb|I|\verb|I|\verb|K|\d\verb|J|\verb|+|\verb|+|\verb|j|\verb|-|\verb|K|\verb|K|\d\verb|K|\verb|i|\verb|J|\verb|J|\verb|J|\verb|J|\verb|i|\d\verb|K|\verb|K|\verb|K|\verb|-|\verb|j|\verb|+|\verb|+|\d\verb|J|\verb|K|\\\phantom{P_{00}=[}\verb|I|\verb|I|\verb|+|\verb|-|\verb|j|\d\verb|I|\verb|i|\verb|I|\verb|k|\verb|i|\verb|j|\verb|J|\d\verb|+|\verb|+|\verb|K|\verb|k|]$

$P_{57}=[\verb|I|\verb|i|\verb|+|\verb|+|\verb|I|\verb|+|\verb|-|\verb|+|\verb|+|\verb|-|\verb|i|\verb|i|\verb|I|\verb|-|\verb|-|\verb|-|\verb|+|\verb|i|\verb|+|\verb|I|\verb|-|\verb|i|\verb|-|\verb|-|\verb|I|\verb|I|\verb|-|\verb|I|\d\verb|+|\verb|I|\verb|-|\verb|I|\verb|I|\verb|-|\verb|-|\verb|i|\verb|-|\verb|I|\verb|+|\verb|i|\\\phantom{P_{00}=[}\verb|+|\verb|-|\verb|-|\verb|-|\verb|I|\verb|i|\verb|i|\verb|-|\verb|+|\verb|+|\verb|-|\verb|+|\verb|I|\verb|+|\verb|+|\verb|i|\verb|I|]$

$P_{58}=[\verb|+|\verb|I|\verb|i|\verb|I|\verb|K|\verb|K|\verb|J|\verb|J|\verb|k|\verb|i|\verb|i|\verb|K|\verb|K|\verb|k|\d\verb|-|\verb|K|\verb|K|\verb|k|\verb|i|\verb|I|\verb|k|\verb|j|\verb|J|\verb|k|\verb|K|\verb|i|\verb|i|\verb|i|\verb|+|\verb|+|\verb|i|\verb|i|\verb|i|\verb|K|\verb|k|\verb|J|\verb|j|\verb|k|\verb|I|\verb|i|\\\phantom{P_{00}=[}\verb|k|\verb|K|\verb|K|\d\verb|-|\verb|k|\verb|K|\verb|K|\verb|i|\verb|i|\verb|k|\verb|J|\verb|J|\verb|K|\verb|K|\verb|I|\verb|i|\verb|I|\verb|+|]$

$P_{60}=[\verb|K|\verb|j|\verb|+|\verb|k|\verb|k|\verb|-|\verb|+|\d\verb|K|\verb|i|\verb|+|\verb|j|\verb|k|\verb|i|\verb|J|\verb|J|\verb|k|\verb|K|\verb|+|\verb|J|\verb|K|\verb|i|\verb|+|\d\verb|+|\verb|I|\verb|I|\verb|j|\verb|J|\verb|I|\verb|k|\verb|J|\verb|J|\verb|k|\verb|I|\verb|J|\verb|j|\verb|I|\verb|I|\d\verb|+|\verb|+|\verb|i|\\\phantom{P_{00}=[}\verb|K|\verb|J|\verb|+|\verb|K|\verb|k|\verb|J|\verb|J|\verb|i|\verb|k|\verb|j|\verb|+|\verb|i|\d\verb|K|\verb|+|\verb|-|\verb|k|\verb|k|\verb|+|\verb|j|\verb|K|]$

$P_{61}=[\verb|K|\verb|k|\verb|K|\verb|K|\verb|I|\verb|i|\verb|K|\verb|i|\verb|K|\verb|I|\verb|I|\verb|k|\verb|k|\verb|K|\verb|i|\verb|I|\verb|I|\verb|K|\verb|I|\verb|I|\verb|K|\verb|k|\verb|i|\verb|K|\verb|K|\verb|i|\verb|i|\verb|i|\verb|I|\verb|K|\d\verb|+|\verb|K|\verb|I|\verb|i|\verb|i|\verb|i|\verb|K|\verb|K|\verb|i|\verb|k|\\\phantom{P_{00}=[}\verb|K|\verb|I|\verb|I|\verb|K|\verb|I|\verb|I|\verb|i|\verb|K|\verb|k|\verb|k|\verb|I|\verb|I|\verb|K|\verb|i|\verb|K|\verb|i|\verb|I|\verb|K|\verb|K|\verb|k|\verb|K|]$

$P_{62}=[\verb|+|\verb|I|\verb|I|\verb|-|\verb|I|\verb|-|\verb|I|\verb|+|\verb|+|\verb|-|\verb|I|\verb|I|\verb|i|\verb|i|\verb|i|\d\verb|+|\verb|I|\verb|i|\verb|I|\verb|I|\verb|i|\verb|-|\verb|-|\verb|+|\verb|i|\verb|-|\verb|i|\verb|-|\verb|i|\verb|I|\verb|-|\verb|-|\verb|I|\verb|i|\verb|-|\verb|i|\verb|-|\verb|i|\verb|+|\verb|-|\\\phantom{P_{00}=[}\verb|-|\verb|i|\verb|I|\verb|I|\verb|i|\verb|I|\d\verb|+|\verb|i|\verb|i|\verb|i|\verb|I|\verb|I|\verb|-|\verb|+|\verb|+|\verb|I|\verb|-|\verb|I|\verb|-|\verb|I|\verb|I|\verb|+|]$

$P_{63}=[\verb|-|\verb|i|\verb|i|\verb|i|\verb|I|\verb|i|\verb|+|\verb|j|\verb|j|\verb|j|\verb|J|\verb|+|\verb|k|\verb|i|\verb|K|\verb|J|\verb|-|\verb|K|\verb|k|\verb|-|\verb|i|\verb|k|\verb|-|\verb|I|\verb|j|\verb|-|\verb|-|\verb|-|\verb|+|\verb|i|\verb|I|\d\verb|-|\verb|I|\verb|i|\verb|+|\verb|-|\verb|-|\verb|-|\verb|j|\verb|I|\\\phantom{P_{00}=[}\verb|-|\verb|k|\verb|i|\verb|-|\verb|k|\verb|K|\verb|-|\verb|J|\verb|K|\verb|i|\verb|k|\verb|+|\verb|J|\verb|j|\verb|j|\verb|j|\verb|+|\verb|i|\verb|I|\verb|i|\verb|i|\verb|i|\verb|-|]$

$P_{64}=[\verb|-|\verb|-|\verb|-|\verb|+|\verb|-|\verb|-|\verb|+|\verb|-|\verb|j|\verb|J|\verb|j|\verb|j|\verb|j|\verb|J|\verb|J|\verb|J|\verb|K|\verb|K|\verb|K|\verb|k|\verb|k|\verb|k|\verb|K|\verb|k|\verb|i|\verb|I|\verb|i|\verb|i|\verb|I|\verb|i|\verb|i|\verb|i|\verb|i|\verb|i|\verb|i|\verb|I|\verb|i|\verb|i|\verb|I|\verb|i|\\\phantom{P_{00}=[}\verb|k|\verb|K|\verb|k|\verb|k|\verb|k|\verb|K|\verb|K|\verb|K|\verb|J|\verb|J|\verb|J|\verb|j|\verb|j|\verb|j|\verb|J|\verb|j|\verb|-|\verb|+|\verb|-|\verb|-|\verb|+|\verb|-|\verb|-|\verb|-|]$

$P_{66}=[\verb|J|\verb|-|\verb|i|\verb|j|\verb|K|\verb|k|\verb|K|\verb|j|\verb|k|\verb|K|\verb|-|\verb|I|\verb|i|\verb|K|\verb|+|\verb|i|\d\verb|+|\verb|I|\verb|+|\verb|k|\verb|i|\verb|i|\verb|-|\verb|k|\verb|k|\verb|J|\verb|K|\verb|K|\verb|K|\verb|J|\verb|i|\verb|+|\verb|J|\verb|J|\verb|+|\verb|i|\verb|J|\verb|K|\verb|K|\verb|K|\\\phantom{P_{00}=[}\verb|J|\verb|k|\verb|k|\verb|-|\verb|i|\verb|i|\verb|k|\verb|+|\verb|I|\d\verb|+|\verb|i|\verb|+|\verb|K|\verb|i|\verb|I|\verb|-|\verb|K|\verb|k|\verb|j|\verb|K|\verb|k|\verb|K|\verb|j|\verb|i|\verb|-|\verb|J|]$

$P_{68}=[\verb|I|\verb|I|\verb|K|\verb|i|\verb|i|\verb|j|\verb|J|\verb|k|\d\verb|K|\verb|j|\verb|j|\verb|k|\verb|i|\verb|+|\verb|k|\verb|+|\verb|I|\verb|+|\verb|I|\verb|J|\verb|I|\verb|-|\verb|j|\verb|k|\verb|K|\d\verb|+|\verb|j|\verb|K|\verb|K|\verb|-|\verb|+|\verb|j|\verb|-|\verb|+|\verb|+|\verb|-|\verb|j|\verb|+|\verb|-|\verb|K|\\\phantom{P_{00}=[}\verb|K|\verb|j|\d\verb|+|\verb|K|\verb|k|\verb|j|\verb|-|\verb|I|\verb|J|\verb|I|\verb|+|\verb|I|\verb|+|\verb|k|\verb|+|\verb|i|\verb|k|\verb|j|\verb|j|\d\verb|K|\verb|k|\verb|J|\verb|j|\verb|i|\verb|i|\verb|K|\verb|I|\verb|I|]$

$P_{69}=[\verb|-|\verb|-|\verb|-|\verb|-|\verb|-|\verb|j|\verb|-|\verb|J|\verb|J|\verb|j|\verb|-|\verb|-|\verb|+|\verb|j|\verb|-|\verb|+|\verb|-|\verb|J|\verb|J|\verb|J|\verb|j|\verb|-|\verb|j|\verb|+|\verb|J|\verb|j|\verb|-|\verb|j|\verb|j|\verb|-|\verb|J|\verb|j|\verb|+|\verb|-|\d\verb|+|\verb|-|\verb|+|\verb|j|\verb|J|\verb|-|\\\phantom{P_{00}=[}\verb|j|\verb|j|\verb|-|\verb|j|\verb|J|\verb|+|\verb|j|\verb|-|\verb|j|\verb|J|\verb|J|\verb|J|\verb|-|\verb|+|\verb|-|\verb|j|\verb|+|\verb|-|\verb|-|\verb|j|\verb|J|\verb|J|\verb|-|\verb|j|\verb|-|\verb|-|\verb|-|\verb|-|\verb|-|]$}

\vspace{0.35cm}
\section*{Acknowledgment}
The authors would like to thank the editors and reviewers for their careful proofreading of our work.


\vspace{-0.35cm}
\begin{IEEEbiographynophoto}{Curtis Bright}
received a PhD in computer science from the University of Waterloo in 2017.  He was a research intern at Maplesoft (2017, 2018) and a postdoctoral fellow at the University of Waterloo (2017--2019) and Carleton University (2020).  He joined the faculty of the School of Computer Science at the University of Windsor in 2020 and starts as a tenure-track professor in 2021.  He is currently the lead developer of the MathCheck project for combining satisfiability solvers with computer algebra systems to solve mathematical conjectures.  His research interests include automated reasoning search algorithms, computer-assisted proofs, experimental mathematics, formal verification or counterexample generation of conjectures, and discrete mathematics.  He proved that Williamson sequences exist in all lengths that are powers of two as a postdoctoral fellow at the University of Waterloo under the supervision of professor Vijay Ganesh and as a research affiliate of Wilfrid Laurier University under the supervision of professor Ilias Kotsireas.
\end{IEEEbiographynophoto}
\vspace{-0.35cm}
\begin{IEEEbiographynophoto}{Ilias Kotsireas}(Member, IEEE)
serves as a professor of Computer Science at Wilfrid Laurier University (Waterloo, Ontario, Canada) and Director of the CARGO Lab. He has over 190 refereed journal and conference publications, chapters in books, edited books and special issues of journals, in the research areas of Computational Algebra, Metaheuristics, High-Performance Computing, Dynamical Systems and Combinatorial Design Theory. He serves on the Editorial Board of 8 international journals. He serves as the Managing Editor of~2 Springer journals and as Editor in Chief of a Springer journal and a Birkhauser book series. He has organized a very large number of international conferences in Europe, North America and Asia, often serving as a Program Committee Chair or General Chair. His research work has gathered approximately 1200 citations. His research is and has been funded by NSERC, the European Union and the National Natural Science Foundation of China (NSFC). He has received funding for conference organization from Maplesoft, the Fields Institute, and several Wilfrid Laurier University offices. He served as Chair of the ACM Special Interest Group on Symbolic Computation (SIGSAM) on a 4-year term: July 2013 to July 2017. He is a co-chair of the Applications of Computer Algebra Working Group (ACA WG), a group of 45 internationally renowned researchers that oversee the organization of the ACA conference series. He has delivered more than 10 invited/plenary talks at conferences around the world.
\end{IEEEbiographynophoto}
\vspace{-0.35cm}
\begin{IEEEbiographynophoto}{Vijay Ganesh}
is an associate professor at the University of Waterloo. Prior to joining Waterloo in 2012, he was a research scientist at MIT (2007--2012) and completed his PhD in computer science from Stanford University in 2007. Vijay's primary area of research is the theory and practice of automated reasoning aimed at software engineering, formal methods, security, and mathematics. In this context he has led the development of many SAT/SMT solvers, most notably, STP, the Z3 string solver, MapleSAT, and MathCheck. He has also proved several decidability and complexity results in the context of first-order theories. He has won over 25 research awards, best paper awards, distinctions, and medals for his research to-date. He recently won an ACM ISSTA Impact Paper Award 2019 (Best Paper in 10 Years Award at ISSTA conference), an ACM Test of Time Award at CCS 2016 (Best Paper in 10 Years Award at CCS), the Early Researcher Award (ERA) given by the Ontario Government in 2016, Outstanding Paper Award at ACSAC 2016, an IBM Research Faculty Award in 2015, two Google Research Faculty Awards in 2013 and 2011, a Ten-Year Most Influential paper citation at DATE 2008, and~10 best paper awards/honors of different kinds at conferences like CAV, IJCAI, CADE, ISSTA, SAT, SPLC, and CCS. His solvers STP and MapleSAT have won numerous awards at the highly competitive international SMT and SAT solver competitions. In 2013 he was invited to the first Heidelberg Laureate Forum, a gathering where a select group of young researchers from around the world met with Turing, Fields and Abel Laureates.
\end{IEEEbiographynophoto}

\end{document}